\documentclass[11pt]{article}
\usepackage{graphicx}
\usepackage{amsfonts}
\usepackage{amsmath}
\usepackage{amssymb}
\usepackage{url}
\usepackage{fancyhdr}
\usepackage{indentfirst}
\usepackage{enumerate}
\usepackage{dsfont,footmisc}
\usepackage{tikz}

\usetikzlibrary{arrows.meta,positioning}
\usepackage{subcaption}

\usepackage{amsthm}
\usepackage{color}
\usepackage{natbib}
\usepackage[british]{babel}
\usepackage[colorlinks=true,citecolor=blue]{hyperref}
\usepackage[right,pagewise,displaymath, mathlines]{lineno}

\usepackage{soul,xcolor}
\setstcolor{red}

\def\d{\mathrm{d}}
\def\laweq{\buildrel \d \over =}

\def\pcn{\buildrel \mathrm p \over \rightarrow}

\newcommand{\VaR}{\mathrm{VaR}}

\newcommand{\U}{\mathrm{U}}
\newcommand{\ES}{\mathrm{ES}}
\newcommand{\E}{\mathbb{E}}
\newcommand{\R}{\mathbb{R}}

\newcommand{\N}{\mathbb{N}}
\newcommand{\p}{\mathbb{P}}

\newcommand{\id}{\mathds{1}}

\newcommand{\X}{\mathcal X}

\renewcommand{\H}{\mathbb H}

\newcommand{\esssup}{\mathrm{ess\mbox{-}sup}}
\newcommand{\essinf}{\mathrm{ess\mbox{-}inf}}
\renewcommand{\ge}{\geqslant}
\renewcommand{\le}{\leqslant}

\renewcommand{\leq}{\leqslant}
\renewcommand{\epsilon}{\varepsilon}

\theoremstyle{plain}
\newtheorem{theorem}{Theorem}
\newtheorem{lemma}{Lemma}
\newtheorem{proposition}{Proposition}

\theoremstyle{definition}
\newtheorem{definition}{Definition}
\newtheorem{example}{Example}

\theoremstyle{remark}
\newtheorem{remark}{Remark}

\theoremstyle{definition}

\renewcommand{\cite}{\citet}

\usepackage[onehalfspacing]{setspace}

\begin{document}

\title{Risk Concentration and the  Mean-Expected Shortfall Criterion}

\author{
Xia Han\thanks{Department of Statistics and Actuarial Science, University of Waterloo, Canada.  E-mail: \texttt{x235han@uwaterloo.ca}}
\and  Bin Wang\thanks{Academy of Mathematics and Systems Science, Chinese Academy of Sciences, Beijing, China. E-mail: \texttt{wangbin@amss.ac.cn}} \and Ruodu Wang\thanks{Department of Statistics and Actuarial Science, University of Waterloo, Canada. E-mail: \texttt{wang@uwaterloo.ca}}
\and  Qinyu Wu\thanks{Department of Statistics and Finance, University of Science and Technology of China, China. E-mail: \texttt{wu051555@mail.ustc.edu.cn}}}

 \date{\today}

\maketitle

\begin{abstract} Expected Shortfall (ES, also known as CVaR) is the most important coherent risk measure in finance, insurance, risk management, and engineering. Recently, \cite{WZ21}   put forward four   economic axioms for portfolio risk assessment  and provide the first economic axiomatic foundation for the family of
$\ES$.   In particular, the axiom of no reward for concentration (NRC) is arguably quite strong, which imposes an additive form of the risk measure on portfolios with a certain dependence structure.
We move away from the axiom of NRC by introducing the notion of {\emph {concentration aversion}}, which does not impose any specific form of the risk measure.   It turns out that risk measures with concentration aversion  are   functions of  ES and the expectation. Together with the other three standard axioms of monotonicity, translation invariance and lower semicontinuity, concentration aversion uniquely characterizes the family of ES.
In addition,  we establish  an axiomatic foundation for  the problem of mean-ES  portfolio selection
and  new explicit formulas for convex and consistent risk measures.
Finally, we provide an economic justification for   concentration aversion via  a few axioms on the attitude of a regulator towards dependence structures. 

\end{abstract}

\section{Introduction}
The quantification of market risk for pricing,  portfolio selection, and  risk management purposes has long been a point of interest to researchers and practitioners in finance.  Since the early 1990s, Value-at-Risk (VaR)  has been the leading tool for measuring market risk because of its  conceptual simplicity and easy evaluation.  It is well known that VaR has been criticized because of its fundamental deficiencies; for instance, it does not account for ``tail risk"  and  its lack of subadditivity or convexity; see e.g.,  \cite{DEGKMRS01}. These limitations have prompted the implementation of an alternative measure of risk, the  Expected Shortfall (ES), also known as CVaR, TVaR and AVaR in various contexts.

 As the dominating class of risk measures in financial practice, ES has many nice theoretical  properties. In particular,  ES satisfies the four axioms of coherence \citep{ADEH99}, and it is also additive for comonotonic risks \citep{K01}, and thus it is
  a convex Choquet integral  \citep{Y87, S89}. In addition
to these theoretic properties, ES  admits a nice representation as the minimum of
expected losses \citep{RU02}, which allows for convenience in convex optimization.
In the recent Fundamental Review of the Trading Book \citep{B16, B19}, the Basel Committee on Banking Supervision proposed a shift from the $99 \%$ VaR to the
$97.5 \%$ ES as the standard risk measure for internal models in market risk assessment. All the above reasons make ES arguably the most important risk measure in banking practice and insurance regulation.

The study of
axiomatic characterization of risk measures  provides guidelines for choosing among various choices of risk measures.
Several sets of axioms
have been established to characterize VaR, including those of \cite{C091}, \cite{KP16}, \cite{HP18}, and \cite{LW21}.  Fewer scholars analyze the axiomatic foundation for ES. In some papers,  ES is identified  based on its joint property with the corresponding VaR; in particular,
ES is  the smallest law-invariant coherent risk measure dominating  VaR \citep{D02}, the only coherent distortion risk measure co-elicitable with  VaR \citep{WW20b}, and the only coherent Bayes risk measure with VaR being its Bayes estimator \citep{EMWW21}.


Different from the above literature relying on VaR to identify ES,   \cite{WZ21} proposed four   axioms, monotonicity, law invariance, prudence and no reward for concentration (NRC), in the context of portfolio risk assessment, which jointly characterize the family of ES.
The key axiom [NRC] means that a \emph{concentrated} portfolio, whose components incur large losses simultaneously in a   stress event $A$ of regulatory interest,  does not receive any capital reduction. This axiom  reflects  two important common features in portfolio risk assessment. The first   is that  regulators are concerned with tail events, which are rare events (i.e., have small probabilities) in which risky positions incur large losses, and  the second    concerns  diversification and risk concentration.
Mathematically, [NRC] is  quite a  strong  property as  it gives the additive form of the risk measure on  concentrated portfolios.  Hence, [NRC] does not apply in contexts where values of the underlying risk measures  are not meant to be additive, such as risk rating or ranking decisions; nevertheless, ES can be used for rating or ranking credit risks, as in, e.g., \cite{GKWW20}.

The main purpose of this paper is the  study  of an alternative, more natural, property which does not impose any specific functional form and can replace [NRC].
This alternative property  will be called   \emph{concentration aversion}
(CA), whose desirability in regulation can be justified  by the   arguments of \cite{WZ21} who extensively discussed issues related to risk concentration and diversification benefit.  Although reflecting similar economic considerations, none of [CA] and [NRC] implies the other. As [CA] is free of any particular functional form,   it is invariant under any strictly increasing transforms on the risk measure, and this invariance is not shared by [NRC].
  In Section \ref{sec:2},  some preliminaries about risk measures are collected, and  the key property [CA] is formulated.  We show that together with law invariance,  [CA] is equivalent to a more mathematically tractable property [$p$-CA] in  Proposition \ref{prop:1}.

As the
 first main result of this paper, Theorem \ref{th-p-CA} in  Section \ref{sec:3} says that  the risk measures satisfying   [$p$-CA] are precisely  functions  of  ES and expectation.  The proof of Theorem \ref{th-p-CA} is quite different from techniques used in \cite{WZ21}, and it requires some novel mathematical tools including a recent advanced result from \cite{WW20}.
 We proceed to illustrate   in Theorem \ref{th:port}  that  [CA]  characterizes the mean-ES criteria in portfolio selection,  thus   providing an axiomatic foundation for such optimization problems.
 The mean-risk  portfolio selection problem has a long history since \cite{M52}; see also \cite{BS01}, \cite{RU02} and the more recent \cite{HK21}.

In Section \ref{sec:4},  we concentrate on  monetary risk measures, the most popular type of risk measures; for a comprehensive treatment, see \cite{FS16}.
It turns out that monetary risk measures satisfying [CA] admit a simple representation as a special type of mean-deviation risk measures (Theorem \ref{th-p}), where the deviation is measured by a transformed difference between ES and the mean.
Quite surprisingly, if we further impose lower semi-continuity, then such a monetary risk measure has to be an ES (Theorem \ref{th-ES}).
Compared to the main result of  \cite{WZ21}, our new characterization enhances  the axiomatic theory for $\ES$ as no particular additive form needs to be assumed ex ante.
Moreover, we obtain characterizations for coherent, convex, or consistent risk measures \citep{MW20} satisfying [CA], giving rise to many new explicit examples of convex and nonconvex consistent risk measures.

In the main part of the paper, the domain of   risk measures of interest is chosen as the set of bounded random variables.  Generalizations and technical remarks related  extending the above results to larger spaces of random variables are discussed in Section \ref{sec:5}. In particular,  all our main results can be readily extended to $L^q$ spaces for $q\ge1$  under a continuity assumption.

Finally, in Section \ref{sec:6-r1}, we provide an endogenous economic reasoning for   concentration aversion via  a few axioms on the attitude of a regulator towards bivariate dependence structures. 
We show in Theorems \ref{th:r1-1} and \ref{th:r1-2} that, the four natural requirements of non-diversifiability, dependence monotonicity, convexity and maximality jointly characterize the dependence structures   modeling risk concentration in this paper. This result provides a theoretical  support to [CA], as well as   [NRC] of \cite{WZ21}, in the  context of regulatory risk measures.
To the best of our knowledge, there is no similar study in the literature on axiomatizing sets of dangerous dependence structures.

 \section{Risk concentration and concentration aversion}\label{sec:2}

Throughout this paper, we work with an atomless probability space $(\Omega,\mathcal F,\p)$. All equalities and inequalities of functionals on $(\Omega,\mathcal F,\p)$ are under $\p$ almost surely ($\p$-a.s.) sense.  A risk measure  $\rho$ is a mapping  from $\X$ to $(-\infty,\infty]$, where $\X$ is  a convex cone of random variables representing losses faced by financial institutions.
For $q\in (0,\infty)$, denote by $L^q=L^q(\Omega,\mathcal F,\p)$ the set of all random variables $X$ with $\E[|X|^q]<\infty$ where $\E$ is the expectation under $\p$. Furthermore, $L^\infty=L^\infty(\Omega,\mathcal F,\p)$ is the space of all essentially bounded random variables, and $L^0=L^0(\Omega,\mathcal F,\p)$ is denoted by the space of all random variables. Positive values of random variables in $\mathcal{X}$ represent one-period losses. We write $X \stackrel{\mathrm{d}}{=} Y$ if two random variables $X$ and $Y$ have the same law. We first collect the key  concepts of tail events and risk concentration as in \cite{WZ21}.

\begin{definition}[Tail events and risk concentration]\label{p-tail-event-a}
Let $X$ be a random variable and $p\in(0,1)$.
\begin{enumerate}[(i)]
\item
A \emph{tail event} of  $X$ is an event $A\in \mathcal F$  with $0<\p(A)<1$  such that   $X(\omega)\ge X(\omega')$ holds for a.s.~all  $\omega\in A$ and $\omega'\in A^c$,
where $A^c$ stands for the complement of $A$.
\item A \emph{$p$-tail event} of $X$ is a tail event of $X$ with probability $1-p$.
\item
A random vector $(X_1,\dots,X_n)$ is \emph{$p$-concentrated} if its components share a common $p$-tail event.
\item A random vector $(X_1,\dots,X_n)$ is \emph{comonotonic}
if there exists a random variable $Z$ and increasing functions $f_1,\dots,f_n$ on $\R$ such that $X_i=f_i(Z)$ a.s.~for every $i=1,\dots,n$.
\end{enumerate}
\end{definition}

The terminology that a $p$-tail event has probability $1-p$ stems from the regulatory language
where, for instance, a tail event with probability $1 \%$ corresponds to the calculation of a $99\%$ VaR.
A random vector  $(X_1,\dots, X_n)$  is $p$-concentrated  for all $p\in(0,1)$ if and only if it is comonotonic; see Theorem 4 of \cite{WZ21}.  Hence, $p$-concentration can be seen as a weaker notion of positive dependence than comonotonicity, which is  a popular notion in the axiomatic characterization of risk functionals and preferences; see  e.g.,
 \cite{Y87} and \cite{S89}.
 For more details  and a real-data example on  $p$-concentration,  see \cite{WZ21}.


Next, we define the two important risk measures in banking and insurance practice. The $\VaR$ at level $p\in (0,1)$ is the functional $\VaR_p:L^0 \to \mathbb{R}$ defined by
$$
 \VaR_p(X)= \inf\{x\in \R: \p(X\le x)\ge  p\},
$$
which precisely is the left $p$-quantile of $X$, and the $\ES$ at  level $p\in(0,1)$ is the functional $\ES_p:L^1 \to \mathbb{R}$ defined by
$$
  \ES_p(X)=\frac{1}{1-p}\int_p^1 \VaR_s(X)\d s.
$$
In this paper, terms such as increasing or decreasing functions are in the non-strict sense.
A few axioms and  properties of a risk measure $\rho$ on $\mathcal X$ are collected below, where all random variables are tacitly assumed to be in the space $\X$.

\begin{enumerate}
\item[{[M]}] Monotonicity: $\rho(X)\le \rho(Y)$ whenever $X\le Y$ (pointwise).
\item[{[TI]}] Translation invariance: $\rho(X+c)=\rho(X)+c$ for all $c\in\R$.
\item[{[LI]}] Law invariance:  $\rho(X)=\rho(Y)$ whenever $X\laweq Y$.
\item[{[P]}] Lower semicontinuity: $\liminf_{n\to\infty} \rho(X_n) \ge \rho(X)$ if $X_n\to X$ (pointwise).
\item[ {[NRC]}] No reward for concentration:
There exists an event $A \in \mathcal F$ such that
$ \rho(X+Y)=\rho(X)+\rho(Y) $ holds for all risks  $X $ and $Y$ sharing the tail event $A$.
\item[{[$p$-TA]}] $p$-tail additivity with $p\in (0,1)$: $\rho(X+Y)=\rho(X)+\rho(Y)$
  for all $X$ and $Y$ sharing a $p$-tail event.\footnote{The property [$p$-TA] is called \emph{$p$-additivity} by \cite{WZ21}.}
\end{enumerate}

\cite{WZ21} proposed [M], [LI], [P] and [NRC] as  four axioms, and showed that they together characterize the class of ES up to scaling; see their Theorem 1 and Endnote 14.\footnote{Lower semicontinuity is called \emph{prudence} by \cite{WZ21} and hence the abbreviation [P].}
 Axioms [M], [TI], and [LI]  are standard in the literature of monetary risk measures. Axiom [P] is motivated by the statistical consideration of robustness (\cite{H71}).  It reflects the idea that if the loss $X$ is statistically modelled using a truthful approximation (e.g., via a consistent distribution estimator), then the approximated risk model should not underreport the capital requirement as the approximation error tends to zero. Therefore,  Axiom [P] is also a natural requirement for a reasonable risk measure  used in practice, as argued by \cite{WZ21}.
The inequality in [M] and convergence in [P] are formulated in  a  pointwise sense,   making these axioms weaker and the corresponding characterization results stronger.
Nevertheless, as discussed in Remark 1 of  \cite{WZ21}, one can replace ``$X\le Y$ (pointwise)" in [M] by ``$X\le Y$ $\p$-a.s." or ``$X\preceq_{\rm st} Y$\footnote{The partial order $X\preceq_{\rm st} Y$ means  $\p(X\le t)\ge \p(Y\le t)$ for all $t\in\R$.}", and replace pointwise convergence  in [P]  by   in probability,  in distribution, or a.s.~convergence.\footnote{Continuity with respect to convergence in distribution is an equivalent formulation    of robustness as shown by \cite{H71}.
A related property in the literature of convex risk measures is the Fatou property,  meaning  that $\liminf _{n \rightarrow \infty} \rho\left(X_{n}\right) \geqslant \rho(X)$ whenever $\left\{X_{n}\right\}_{n \in \mathbb{N}}$ is a bounded sequence in $\mathcal{X}$ converging to $X$ (pointwise). Clearly, [P] is stronger than the Fatou property, and we will see in Theorem \ref{th-ES} that [P] is used to the characterization of ES. The Fatou property is essential to a dual representation of convex risk measures, and we do not assume convexity in most results in our paper. }
All results in this paper would still hold with the above modified versions.

 As discussed above,  [{\rm NRC}] intuitively means that a concentrated portfolio, whose components incur large losses simultaneously in the stress event $A$, does not receive any diversification benefit.  For a law-invariant risk measure, the property [NRC] is equivalent to [$p$-TA] for some $p\in(0,1)$; see Proposition 4 of \cite{WZ21}. Thus, it suffices to work with [$p$-TA] when analyzing  the property [NRC] of law-invariant risk measures.   As [NRC]  imposes an additive form for the risk measure evaluated on  concentrated portfolios,  it may be seen as a quite strong  property mathematically, and it cannot be used in a context,  such as rating or ranking risks, where values of risk measures or preference functionals are not interpreted as additive units.
Therefore,  finding an alternative property, without the additive form, that may replace [NRC] to characterize ES (and preferences induced by ES) becomes a  natural  problem.\footnote{We thank Martin Herdegen for raising this question during a seminar at the University of Warwick in October 2020.}

To address this problem, we propose the property \emph{concentration aversion} (CA), in a way similar to [NRC] but without imposing additivity.  Instead of assuming that the risk measure is additive for concentrated portfolios,  the new property of [CA] requires that the risk measure (or decision maker) assigns a larger or equal value
for concentrated portfolios, compared to a portfolio that is not concentrated and otherwise identical.



\begin{definition}\label{def:CA} A risk measure $\rho$ satisfies \emph{concentration aversion} if
 there exists an event $A \in \mathcal{F}$ with $\p(A)\in (0,1)$ such that $\rho(X+Y)\le \rho(X'+Y')$ if  $X\laweq X'$, $Y\laweq Y'$, and $X'$ and $Y'$ share the tail event $A$. This property is denoted by [CA].
\end{definition}

 The  event $A$ in [CA] should be interpreted as a stress event of interest to the regulator.  We will show in Proposition \ref{prop:1} that the specification of $A$ does not matter in characterization results on law-invariant risk measures, and only the probability of $A$ is relevant. A similar observation can also be found in \cite{WZ21}.
The property [CA] has a straightforward preference interpretation; that is,
with marginal distributions fixed, the decision maker prefers non-concentrated portfolios over concentrated ones.  Similarly to [NRC], the desirability of [CA] for regulatory risk measures depends on whether one agrees that $p$-concentration represents a dangerous dependence structure of regulatory concern. This issue has been  discussed by \cite{WZ21} in detail; see also \cite{B19} for evidence and considerations in regulatory practice.
In Section \ref{sec:6-r1}, we will provide a first axiomatic justification of [CA] from a few natural properties on the set of adverse dependence structures, and thus showing that [CA] (or [NRC]) can be motivated endogenously.

None of [CA] and [NRC] implies the other one, although they are closely related. For instance, $X\mapsto \exp(\E[X])$ satisfies  [CA] but not [NRC], whereas $X\mapsto -\ES_p(X)$ satisfies [NRC] but not [CA].

 Many characterization axioms in the literature, including [$p$-TA] and [NRC], compare  the value of a risk measure applied to a portfolio with a combined value of the risk measure applied to individual risks. For instance, subadditivity means that a merger does not create extra risk \citep[]{ADEH99}, convexity means diversification does not increase risk level \citep[]{FS02},
and comonotonic additivity means that a comonotonic portfolio does not receive any risk reduction \citep[]{K01, MM04}.
In contrast, the property [CA] is defined by comparing two portfolios, not comparing values of the specific risk measure; thus this property is free of the specific functional form.
For instance, if $\rho$ satisfies [CA], then so is $f\circ \rho$ for any increasing function $f$; such a feature is not shared by the above properties in the risk measure literature, although it widely appears in the literature of decision theory.

Similar to the translation between [NRC] and [$p$-TA], the property
[CA] can also be translated to a mathematical property that is easier to analyze.
This property, called \emph{$p$-concentration aversion} [$p$-CA], will be the central property analyzed in this paper.

\begin{definition}\label{def:2}
Let $p\in(0,1)$. A risk measure $\rho$ satisfies \emph{$p$-concentration aversion}  if $\rho(X+Y)\le \rho(X'+Y')$ for all  $(X,Y)$ and $p$-concentrated $(X',Y')$  satisfying $X\laweq X'$ and $Y\laweq Y'$.
This property is denoted by [$p$-CA].\end{definition}
We first verify that [CA] can be replaced by [$p$-CA] for some $p\in (0,1)$ in our subsequent analysis.
\begin{proposition}\label{prop:1}
For a risk measure  $\rho$ on $\mathcal X$,
the following are equivalent.
\begin{itemize}
\item[(i)] $\rho$ satisfies {\rm [LI]} and  {\rm [CA]}.
\item[(ii)] $\rho$ satisfies   {\rm[}$p$-{\rm CA}{\rm]} for some $p \in(0,1)$.
\end{itemize}
\end{proposition}
\begin{proof}
``(ii) $\Rightarrow$ (i)": We first show  [$p$-CA] implies [LI].  Let $Y=Y'=0$.  Take identical distributed $X$ and $X'$,  and note that  $(X',Y')$ is $p$-concentrated since $Y'$ is a constant. Property  [$p$-CA] implies that $\rho(X)=\rho(X+Y)\le \rho(X'+Y')=\rho(X')$, and exchanging the positions of $(X,Y)$ and $(X',Y')$ we also have $\rho(X')=\rho(X'+Y')\le \rho(X+Y)=\rho(X)$. Therefore,  $\rho$ is law invariant. To verify [CA], take any event $A$ with probability
$1-p$, and it is straightforward that $\rho$ satisfies [CA] with $A$ being the stress event.

 ``(i) $\Rightarrow$ (ii)": Suppose that $\rho$ satisfies [CA] with $A$ being the
stress event, and let $p=1-\mathbb{P}(A)$. Let $X,Y \in\X$ be two random variables which share a tail event
$B$ of probability $1-p.$
It suffices to show that for  any $\widetilde{X}, \widetilde{Y} \in \mathcal{X}$ with $\widetilde{X} \laweq X $ and  $ \widetilde{Y} \laweq Y$, we have  $
\rho(\widetilde X+\widetilde Y)\leq \rho(X+Y)$.
Similar to the proof of  Proposition 4 in \cite{WZ21}, we   construct two random variables  $X^{\prime}, Y^{\prime}\in\X$   such that  $X'$ and $Y'$ share the same  tail event $A$,  and $\left(X^{\prime}, Y^{\prime}\right)$ and $(X, Y)$ are identically distributed. Using   [LI], we have $\rho\left(X^{\prime}+Y^{\prime}\right)=\rho(X+Y)$.  It then follows from [CA],  $\widetilde{X} \laweq X' $ and  $ \widetilde{Y} \laweq Y'$ that   $$
\rho(\widetilde X+\widetilde Y)\leq\rho\left(X^{\prime}+Y^{\prime}\right)=\rho(X+Y),
$$
which completes the proof.\end{proof}

It is immediate from Theorem 5 of \cite{WZ21} that $\ES_p$ satisfies [$p$-CA].
Moreover, the mean $\E$ and  convex combinations of $(\E,\ES_p)$ such as $\lambda \E+ (1-\lambda)\ES_p$ for $\lambda \in (0,1)$ also satisfy [$p$-CA].
For  applications in regulatory risk assessment, the value of $p$ should  be close to $1$, indicating an emphasis on tail events with large losses that happen with a small probability. In \cite{B19}, the choice of $p$ in $\ES_p$ is $0.975$.

 In the following sections, we will formally study risk measures with the property of [$p$-CA]; equivalently, they are law invariant risk measures satisfying [CA].

\begin{remark}
The property [$p$-CA] is defined for an arbitrary but fixed $p\in(0,1)$.
If we allow $p$ to take value $0$, then [$p$-CA] in Definition \ref{def:2}  degenerates to the property that $\rho(X+Y)=\rho(X'+Y')$ for any $X\laweq X'$ and $Y\laweq Y'$.
Such a property is called   \emph{dependence neutrality} by \cite{WW20}, who showed that this property is only satisfied by a transformation of the mean.
\end{remark}

\section{Concentration aversion characterizes mean-ES criteria}\label{sec:3}\label{sec:3}
In this section, we present our first main result  that the property [$p$-CA] characterizes the class of functionals that are transformations of ${\rm{ES}}_p$ and the mean.
In this  and the next sections, we assume that $\rho$ is a risk measure on $\mathcal{X}=L^\infty$, which is the standard choice in the risk measure literature \citep[see e.g.,][]{FS16}.
The extension
of $\mathcal{X}$ to more general spaces  will be  discussed in  Section \ref{sec:5}.

\subsection{Two technical lemmas}

We first collect two lemmas that will become useful tools in the proof of our main result.
Denote by $F_X$   the distribution function of a random variable $X$. 
 Let $F_X^{-1}$ be the left quantile function of $X$, i.e., $$F_X^{-1}(p)={\rm VaR}_p(X)=\inf\{x: F_X(x)\ge p\}.$$ Noting that the probability space is atomless,  there exists a uniform random variable $U$ on $[0,1]$ such that  $F_X^{-1}(U)=X$ a.s.; see e.g., Lemma A.32 of \cite{FS16}. Denote by $\essinf X $ and $\esssup X$ the essential infimum and essential supremum of a random variable $X$, respectively.  Moreover,
define
$$
 L(F_X)=\esssup  X -\essinf X=F_X^{-1}(1-)-F_X^{-1}(0+),
$$
and  let $ T(F_X)$ be the distribution of $F_X^{-1}(U)/2+F_X^{-1}(1-U)/2$ for $U\sim {\rm U}(0,1)$.
The first lemma below discusses the relationship between $L(F_X)$ and $L(T(F_X))$. The second lemma of \cite{WW20} is highly nontrivial, which gives the existence of identically distributed random variables whose difference is a pre-specified random variable with mean $0$.

\begin{lemma}\label{lem:RW}
We have $L(F_X) \ge 2 L(T(F_X)) $ for any random variable $X\in L^\infty$.
\end{lemma}
\begin{proof}
Write $Y=F_X^{-1}(U)/2+F_X^{-1}(1-U)/2$.
It is easy to verify that
$$   \frac{  F_X ^{-1}( 0+ ) + F_X ^{-1}( 0.5 )}{2}   \le  Y \le \frac{  F_X ^{-1}( 0.5 ) + F_X ^{-1}( 1- )}{2}. $$
Hence,
$$ \essinf X  \le  \frac{  F_X ^{-1}( 0+ ) + F_X ^{-1}( 0.5 )}{2}   \le  \essinf Y \le  \esssup Y  \le \frac{  F_X ^{-1}( 0.5 ) + F_X ^{-1}( 1- )}{2}  \le \esssup X. $$
As a consequence, we obtain
 $$
   L(T(F_X)) \le  \frac{ \esssup X - \essinf X  }{2} =  \frac{L(F_X)}{2},
 $$
thus showing the lemma.
\end{proof}
\begin{lemma}[Lemma 1 of \cite{WW20}]\label{lm-WW}
For  a random variable $X$ with $\E[X]=0$ , there exist identically distributed random variables $V$ and $V'$ such that $ V - V' \laweq X$ and $L(V) = L(V')\le L(X)$.
\end{lemma}

\subsection{The main characterization result}

We are now ready to present our main result in this section on the characterization of functionals satisfying [$p$-CA].
In what follows, we denote by $\H$ the half-space $\{(x,y)\in\R^2: x\ge y\}$.
The proof of Theorem \ref{th-p-CA} requires sophisticated constructions of many random variables, utilizing both Lemmas \ref{lem:RW}  and \ref{lm-WW}.

\begin{theorem}\label{th-p-CA}
Let $p\in (0,1)$ and $\rho: L^\infty\to(-\infty,\infty]$. The following two statements hold.
\begin{itemize}
\item[(i)]
$\rho$ satisfies {\rm[}$p$-{\rm CA}{\rm]} if and only if it has the form
$f({\rm{ES}}_p,\E)$, where $f:\H\to (-\infty,\infty]$ is  increasing in its first argument.
\item[(ii)]
$\rho$ satisfies {\rm[M]} and {\rm[}$p$-{\rm CA}{\rm]} if and only if it has the form
$f({\rm{ES}}_p,\E)$, where $f:\H\to (-\infty,\infty]$ is  increasing in both arguments.
\end{itemize}
\end{theorem}

\begin{proof}
(i) The sufficiency statement follows from the fact that $\ES_p$ takes its largest possible value for a $p$-concentrated portfolio among all portfolio vectors with given marginal distributions.  To be specific, by  Theorem 5 of \cite{WZ21},  $\left(X_{1},  X_{2}\right)$ is $p$-concentrated if and only if  $\left(X_{1},  X_{2}\right)$ maximizes the $\mathrm{ES}_{p}$ aggregation; that is, for all  $(X',Y')$ and  $p$-concentrated $(X,Y)$  satisfying $X\laweq X'$ and $Y\laweq Y'$, one has ${\ES}_p(X'+Y')\le{\ES}_p(X+Y)$  and  $\E[X'+Y']=\E[X+Y]$. Hence, if $\rho$ is of the form  $f({\rm{ES}}_p,\E)$ and $f$ is  increasing in its first argument, we have $\rho(X'+Y')\le \rho(X+Y)$, which implies that $\rho$
 satisfies [$p$-CA].

 We now prove the necessity statement.  First,  it is clear that [$p$-CA] implies that
$ \rho(X+Y) = \rho(X_1 + Y_1) $ if  $(X,Y)$ and $(X_1,Y_1)$ are both $p$-concentrated and $ X \laweq X_1 \in L^{\infty} $, $ Y \laweq Y_1 \in L^{\infty} $.

For any $Z \in L^{\infty}$, denote  by $m = \VaR_p(Z) = F_Z^{-1}(p)$, and by $a$ and $b$ two constants such that
$$a = \ES_p (Z) = \frac{1}{1- p} \int_p^1 F_Z^{-1} (t) \d t   , ~~~  b=\ES^{-}_p(Z) := \frac{1}{p} \int_0^p F_Z^{-1} (t) \d t .$$
Note that   $ \E[Z]= (1-p) a + pb $. We aim to prove $\rho(Z) = \rho(Z^*)$ where $ Z^* \sim (1-p) \delta_a + p \delta_b  $, which justifies that $\rho(Z)$ is  determined only by values of $\ES_p(Z)$ and $\E[Z]$.

There is nothing to show if $a=b$, which implies that $Z$ is a constant and thus $Z=Z^*$ a.s. We will assume $a>b$ in what follows.

Denote by $G,H$ the distribution functions of $ F_Z^{-1} (U_1)$ and $ F_Z^{-1}(U_2)$, respectively, where $ U_1 \sim \U[ p,1 ] , U_2 \sim \U[0,p] $. Then we write $F_Z = (1-p) G + p H $
with  $  \mathrm{Support} (G) \subseteq [m,  \esssup (F_Z) ]$ and $\mathrm{Support} (H) \subseteq [ \essinf (F_Z) , m ] $.
Take $U \sim \U[0,1]$ and define
$$  X = Y = X_1 = \frac{F_Z^{-1}(U) }{2} ~~\text{and}~~
Y_1 =\left \{
\begin{aligned}
\frac{F_Z^{-1} (p-U) }{2} ~~ &\text{if}~ U < p,  \\
\frac{F_Z^{-1} (1 + p - U) }{2} ~~ &\text{if}~ U > p.
\end{aligned}\right.$$
We can verify that  $(X,Y,X_1,Y_1)$ is $p$-concentrated with common $p$-tail event $\{ U > p \}$, $X\laweq X_1$ and $Y\laweq Y_1$. Moreover, by letting $Z_1=X_1+Y_1$, we have
 $$
 X+Y =  F_Z^{-1} (U)  \laweq Z ~ ~\text{and}~~  F_{Z_1} = (1-p) T(G) + p T(H).
 $$
 Note that $\ES_p(Z_1)=\ES_p(Z)$ and $\E[Z_1]=\E[Z]$.
Properties [$p$-CA] and [LI] lead to $ \rho(Z) = \rho(Z_1)$. By Lemma \ref{lem:RW}, we further obtain
$$
\esssup (Z_1) - F^{-1}_{Z_1}(p+)=L(T(G))\le \frac{L(G)}{2}=\frac{1}{2}  ( \esssup (Z) - F^{-1}_{Z}(p+)),
$$
and
$$
F^{-1}_{Z_1}(p) -  \essinf (Z_1)=L(T(H))\le \frac{L(H)}{2}=\frac{1}{2}  ( F^{-1}_Z(p) -  \essinf (Z)).
$$
%
We repeat the above argument to construct $Z_2$ with $Z_1$ replacing the position of $Z$.
Take any  $\epsilon \in (0, ({a-b})/{4})$. For   large enough $n$  (more precisely, $n\ge \log_2 (  {L(F_Z)}/{ \epsilon} )$), we have
$$  \esssup (Z_n)   -    F^{-1}_{Z_n}(p+) < \epsilon ~~ \text{and} ~~  F^{-1}_{Z_n}(p) -  \essinf (Z_n)   < \epsilon.$$ Combining with  $ \ES_p (Z_n) =  a$ and $ \ES^{-}_p(Z_n) = b ,$
it  then follows that
$$  \p ( a - \epsilon < Z_n < a + \epsilon ) = 1- p  ~~ \text{and} ~~   \p ( b - \epsilon < Z_n < b + \epsilon ) = p.  $$ Note that the above construction preserves the value of $\rho$, that is,
$$ \rho (Z) = \rho(Z_1) = \rho(Z_2) = \dots = \rho(Z_n). $$
Denote by $G_n$ and $H_n$ the distribution functions of $F_{Z_n}^{-1} (U_1)$ and $ F_{Z_n}^{-1}(U_2)$, respectively.
It follows that $F_{Z_n}=(1-p)G_n+pH_n$. Moreover, the mean of  $ G_n $ is $ a$ and the mean of  $ H_n $ is $b $, and
$$ \mathrm{Support} (G_n ) \subseteq (a-\epsilon , a + \epsilon) ~~ \text{and} ~~ \mathrm{Support} (H_n ) \subseteq (b-\epsilon , b + \epsilon). $$

Note that in an atomless probability, there exists a random vector with any specified distribution (e.g., Lemma D.1 of \cite{VW21}).
We take a random vector $(\id_A,V_1,\dots,V_4)$ such that $A$ is an event  independent  of $ (V_1,V_2,V_3,V_4) $ satisfying $\p(A) = 1-p$,
and $(V_1,\dots,V_4)$, whose existence is justified by Lemma \ref{lm-WW}, satisfies
$$ V_1 \laweq V_2, ~ V_1 - V_2 + a \sim G_n ~\text{and}~ \mathrm{Support} (V_1) = \mathrm{Support} (V_2) \subseteq [- \epsilon , \epsilon],$$
$$ V_3 \laweq V_4,~ V_3 - V_4 + b \sim H_n ~\text{and}~ \mathrm{Support} (V_3) = \mathrm{Support} (V_4) \subseteq [- \epsilon,  \epsilon].$$
Define
$$ X = \id_A \left(V_1 + \frac{a}{2}\right) + \id_{A^c} \left( V_3 + \frac{b}{2} \right) , ~~  Y = \id_A  \left(- V_2 + \frac{a}{2}   \right) + \id_{A^c} \left( - V_4 + \frac{b}{2} \right), $$
$$  X^* = \id_A \left(V_1 + \frac{a}{2}\right) + \id_{A^c} \left( V_3 + \frac{b}{2} \right) ,~~ Y^* = \id_A  \left(- V_1 + \frac{a}{2}   \right) + \id_{A^c} \left( -V_3 + \frac{b}{2} \right). $$
Since $  |V_1|,|V_2|,|V_3|,|V_4| \le \epsilon < ({b-a})/{4} $, for any $\omega \in A$ and $ \omega' \in A^c$,  we have
$$  X(\omega) = V_1(\omega) + \frac{a}{2} >  \frac{a}{2} - \epsilon > \frac{b}{2} + \epsilon  >  V_3 (\omega') + \frac{b}{2} = X(\omega').$$
Similarly, we have $  Y(\omega) > Y(\omega')$, $Y^*(\omega) > Y^*(\omega')$ and $X^*(\omega) > X^*(\omega')$. Hence, $ (X,Y,X^*,Y^*) $ is $p$-concentrated with common $p$-tail event $A$, and $ X = X^*$, $ Y \laweq Y^*$, $ X+Y \laweq Z_n $.  Therefore,
$$ \rho(Z) =  \rho(Z_n) = \rho(X+Y) = \rho(X^* + Y^*) =  \rho ( \id_A \times a + \id_{A^c} \times b ) = \rho(Z^*),$$
thus showing the desirable statement   that the value  of $\rho(X)$ only depends on ${\rm ES}_p(X)$ and $\E[X]$, that is, $\rho$ has the form $f({\rm ES}_p,\E)$.

 It remains to prove that the function $x\mapsto f(x,y)$ is increasing for each fixed $y\in\R$.
Suppose $0<p\le 1/2$, and let $X\sim p\delta_{-(1-p)a}+(1-p)\delta_{pa}$ and $Y\sim (1-p)\delta_{-pb+y}+ p \delta_{(1-p)b+y}$ with $0\le b\le a $. For $U\sim {\rm U}(0,1)$, take $X_1=X_2=F_X^{-1}(U)$, $Y_1=F_Y^{-1}(U)$ and $Y_2=F_Y^{-1}(1-U)$. By straightforward calculation, we obtain
$$
\E[X_1+Y_1]=\E[X_2+Y_2]=y,~~~{\rm ES}_p(X_1+Y_1)=p(a-b)+ \frac{pb}{1-p}+y,~~{\rm ES}_p(X_2+Y_2)=p(a-b)+y.
$$
Note that $(X_1,Y_1)$ is $p$-concentrated and $X_1\laweq X_2$, $Y_1\laweq Y_2$. Hence, by using [$p$-CA],
we obtain
$$
f\left(p(a-b)+ \frac{pb}{1-p} +y,y\right)=\rho(X_1+Y_1)\ge \rho(X_2+Y_2)=f\left(p(a-b)+y,y\right).
$$
Since $a-b\ge 0$ and $b\ge 0$  can be arbitrarily chosen, we have that $x\mapsto f(x,y)$ is increasing for each $ y$. Using   similar arguments, monotonicity also holds for $1/2<p<1$. Hence, we complete the proof of (i).

(ii) The sufficiency statement is straightforward. To show the necessity statement, based on the result in (i), it remains to show that  [{\rm M}] implies the monotonicity of  the function $y\mapsto f(x,y)$. Take $A\in\mathcal F$ with probability $1-p$. Define two random variables $X$ and $Y$ such that   $X(\omega)=Y(\omega)=x$ for $\omega\in A$, and $X(\omega)=x_1$, $Y(\omega)=x_2$ for $\omega\in A^c$, where  $x_1\le x_2\le x$. Obviously, we have $X\le Y$,  and it follows that
\begin{align*}
f(x,(1-p)x+px_1)& =f(\ES_p(X),\E[X])\\& =\rho(X)\le \rho(Y)=f(\ES_p(Y),\E[Y])=f(x,(1-p)x+px_2).
\end{align*}
The monotonicity follows from the fact that $x_1\le x_2\le x$ can be arbitrarily chosen.
\end{proof}

\begin{remark}
The functional $\ES_p^-$ is used in the proof  of Theorem \ref{th-p-CA}, but not in its statement.
There is a linear relationship between $\ES_p$, $\ES_p^-$ and $\E$, that is,
$$
p\ES_p^-(X)+(1-p)\ES_p(X)=\E[X].
$$
Therefore, the form $f(\ES_p,\E)$ of the risk measure in Theorem \ref{th-p-CA} can also be represented as $f_1(\ES_p^-,\E)$ or $f_2(\ES_p,\ES_p^-)$ with different conditions on $f_1$ and $f_2$.
\end{remark}

\subsection{Mean-ES portfolio selection}

There is a large literature on mean-risk portfolio selection since \cite{M52} who measured risk by using variance. In the more recent literature,    risk is often measured by a risk measure, such as VaR   \citep[]{BS01, GP05}, ES \citep[]{RU00, RU02, ESW21}, or expectiles \citep[]{B14, L21}. For a recent work on mean-$\rho$ optimization where $\rho$ is a coherent risk measure, see \cite{HK21}.

Remarkably, Theorem \ref{th-p-CA}  gives rise  to an axiomatic foundation for the mean-ES portfolio selection.
Consider a classical optimization problem
\begin{equation}\label{eq:rw2}
\min_{\mathbf a \in A} \mathcal V(g(\mathbf X,\mathbf a))
\end{equation}
where $A$ is a set of possible actions, $\mathcal V:\X\to (-\infty,\infty]$ is an objective functional, $\mathbf X$ is the underlying $d$-dimensional risk vector,
and $g: \R^d \times A \to \R$ is a function representing  the portfolio value. Constraints on the optimization problem can be incorporated into either $A$ or $\mathcal V$.
For instance, one may set $\mathcal V$ to be $\infty$ for positions that violate certain constraints, as we will see below.

We say that the optimization problem \eqref{eq:rw2}
is \emph{a mean-$\rho$ optimization}  for some risk measure $\rho$,
if $\mathcal V$ is determined by $\E$ and $\rho$ and increasing in both. There are two classic versions of mean-$\rho$ optimization problems:
\begin{enumerate}[(a)]
\item Maximizing expected return with a target risk $r\in \R$, that is
\begin{equation}\label{eq:rw1}
\max_{\mathbf a\in A} \E[ - \mathbf a^{\rm T} \mathbf X ]~~~ \mbox{subject to $\rho(\mathbf a^{\rm T} \mathbf X )\le r$},
\end{equation}
where $\mathbf X$ is the vector of losses (negative returns) from individual assets and $A$ is a subset of $\R^d$; recall that $- \mathbf a^{\rm T} \mathbf X $  represents the future portfolio wealth.
By choosing $$\mathcal V(X) =  \E[X]\id_{\{\rho(X)\le r\}} + \infty \id_{\{\rho(X)>r\}}~~\mbox{and}~~g(\mathbf X,\mathbf a)=\mathbf a^{\rm T} \mathbf X$$ with the convention $\infty\times 0=0$,
\eqref{eq:rw1} becomes \eqref{eq:rw2}, which is clearly a mean-$\rho$ optimization.

\item Minimizing risk with a target expected return $u\in \R$, that is,
\begin{equation}\label{eq:rw3}
\min_{\mathbf a\in A} \rho( \mathbf a^{\rm T} \mathbf X )~~~ \mbox{subject to $\E[-\mathbf a^{\rm T} \mathbf X ]\ge u$}.
\end{equation}
This time, by choosing $$\mathcal V(X) =  \rho (X)\id_{\{\E[X]\le -u\}} + \infty \id_{\{\E[X]>-u\}}~~\mbox{and}~~g(\mathbf X,\mathbf a)=\mathbf a^{\rm T} \mathbf X,$$
we arrive again at \eqref{eq:rw2}.
\end{enumerate}
Using Theorem \ref{th-p-CA}, we obtain a characterization of mean-ES   (i.e., mean-$\ES_p$ for some $p\in(0,1)$) optimization problems, which include the classical problems \eqref{eq:rw1} and \eqref{eq:rw3} with $\rho=\ES_p$.

\begin{theorem}\label{th:port}
 An optimization problem \eqref{eq:rw2} is a mean-ES optimization if and only if  its objective $\mathcal V$ satisfies {\rm [M]} and {\rm [}$p$-{\rm CA}{\rm ]} for some $p\in (0,1)$.
\end{theorem}

Theorem \ref{th:port}   illustrates that a preference for dependence (i.e., [$p$-CA]) can help to pin down the particular form of  optimization problems, in addition to characterizing risk measures.
In the next section, we continue to explore the relationship between concentration aversion and characterizing risk measures.

\section{Monetary risk measures satisfying [CA]}\label{sec:4}
In this section, we again assume that $\rho$ is a risk measure on $\mathcal{X}=L^\infty$, and further investigate    \emph{monetary risk measures} satisfying [$p$-CA].
A \emph{monetary risk measure} is a risk measure satisfying {\rm  [M]  and {\rm [TI]}; see \cite{FS16}.
 It is well known that monetary risk measures are one-to-one corresponding to acceptance sets.
An \emph{acceptance set} $\mathcal A$ is a subset of $\X$ which is generated by some monetary risk measure $\rho$ via $\mathcal A=\{X\in \mathcal X: \rho(X)\le 0\}$.
Also note that a monetary risk measure $\rho$ is finite on $L^\infty$ as long as it is finite at some $X\in L^\infty$. Therefore, we can safely assume $\rho:L^\infty\to\R$ in this section.

\subsection{Concentration-averse monetary risk measures}
Let us first recall the definition of \emph{second-order stochastic dominance} (SSD). We say that $X$ is second-order stochastically dominated by $Y$, denoted by $X\preceq_{\rm SSD}Y$, if $\E[u(X)]\le \E[u(Y)]$ for all increasing convex functions $u$.\footnote {SSD is also known as increasing convex order in probability theory and stop-loss order in actuarial science.}
We collect two properties  from \cite{MW20}.
\begin{enumerate}
 \item[{[SC]}] SSD-consistency: $\rho(X)\le \rho(Y)$ whenever  $X\preceq_{\rm SSD}Y$.
 \item[{[DC]}] Diversification consistency: $\rho(X+Y)\le\rho(X^c+Y^c)$  whenever $X\laweq X^c$, $Y\laweq Y^c$ and $(X^c,Y^c)$ is comonotonic.
\end{enumerate}
The property [SC] is often called \emph{strong risk aversion} for a preference functional  \citep{HR69, RS70}, while [DC] is called \emph{comovement aversion} \citep{WW20}. By using a risk
measure satisfying [SC], a financial institution makes decisions that are consistent with the common notion of risk aversion and, in particular, favours a risk with small variability over one with a large variability.
\cite{MW20} showed that, for a monetary risk measure, [SC] and [DC] are equivalent, and they called
monetary risk measures satisfying [SC] \emph{consistent risk measure}, which have a  representation based on ES; see their Theorem 3.1.
Since $p$-concentration is weaker than comonotonicity,  [$p$-CA] implies [DC], and hence a monetary risk measure satisfying [$p$-CA] is automatically a consistent risk measure. In the following theorem, a representation of such a risk measure is established.
This result leads to a   class of risk measures   \eqref{eq-p} that is new to the literature.
In what follows, we say that a  real-valued function $g$ satisfies the $1$-Lipschitz condition if
 \begin{align}\label{eq-condition}
|g(x)-g(y)|\le |x-y|~~\mbox{ for}~x,y \mbox{~in the domain of $g$}.
\end{align}

\begin{theorem}\label{th-p}
Let $p\in (0,1)$ and $\rho$ be a risk measure on $L^\infty$. Then,
$\rho$ satisfies {\rm [M]}, {\rm [TI]}, {\rm [}$p$-{\rm CA}{\rm ]} and $\rho(0)=0$ if and only if it has the form
\begin{align}\label{eq-p}
\rho(X)=g(\ES_p(X)-\E[X])+\E[X],
\end{align}
for some increasing function $g: [0,\infty)\to \R$  with $g(0)=0$ satisfying the $1$-Lipschitz  condition.
In particular, such $\rho$ is a consistent risk measure.
\end{theorem}

\begin{proof}

Let us first prove sufficiency. Obviously, $\rho$ of the form \eqref{eq-p} satisfies [TI] and $\rho(0)=0$. By Theorem \ref{th-p-CA}, we obtain that $\rho$ satisfies [$p$-CA]. So it remains to verify that $\rho$ is monotone. Suppose $X\le Y$, and define $a_1=\ES_p(X)$, $b_1=\E[X]$, $a_2=\ES_p(Y)$ and $b_2=\E[Y]$. Obviously, we have $a_1\le a_2$, $b_1\le b_2$ and
$$
\rho(X)=g(a_1-b_1)+b_1,~~~\rho(Y)=g(a_2-b_2)+b_2.
$$
If $a_2-a_1\ge b_2-b_1$, we have
\begin{align*}
\rho(X)&\le \rho(X)+(b_2-b_1)\\
&=g(a_1-b_1)+b_2\\
&=g((a_1+b_2-b_1)-b_2)+b_2\\
&\le g(a_2-b_2)+b_2=\rho(Y),
\end{align*}
where the second inequality follows from the increasing monotonicity of $g$.
If $a_2-a_1< b_2-b_1$, we have
\begin{align*}
\rho(X)&\le \rho(X)+(a_2-a_1)\\
&=g(a_2-(b_1+a_2-a_1))+(b_1+a_2-a_1)\\
&\le \rho(a_2-b_2)+b_2=\rho(Y),
\end{align*}
where the second inequality follows from the 1-Lipschitz condition of $g$. Hence, we complete the proof of sufficiency. For the other direction, it follows from the results in Theorem \ref{th-p-CA} that $\rho$ has the form $f(\ES_p,\E)$ for some bivariate function $f$.
Define a function $g: [0,\infty)\to  \R$ such that $g(x)=f(x,0)$ for $x\ge 0$. It is clear that $g(0)=f(0,0)=\rho(0)=0$.
Note that $f(\cdot, y): [y,\infty)\to \R$ is increasing for all $y\in\R$ (see Theorem \ref{th-p-CA}). It follows that $g$ is increasing.
Using [TI], we obtain
\begin{align*}
\rho(X)&=\rho(X-\E[X])+\E[X]\\
&=f(\ES_p(X)-\E[X],0)+\E[X]\\
&=g(\ES_p(X)-\E[X])+\E[X].
\end{align*}
Finally, applying Theorem \ref{th-p-CA} (ii), we know that the function $f(x,\cdot): (-\infty,x]\to\R$ is increasing for all $x\in\R$. Hence, we have
$$
g(x-y)+y=f(x,y)\le f(x,y')=g(x-y')+y'~~{\rm for~all}~y<y'\le x,
$$
which implies that  $g$ is 1-Lipschitz. Hence, we complete the proof.
\end{proof}

The risk measure $\rho$ with form \eqref{eq-p}
is the sum of the mean and $g(\ES_p-\E)$.
Note that  $\ES_p-\E$ is both a generalized deviation measure  according to   \cite{RUZ06}
and  a coherent  measure of variability according to \cite{FWZ17}.
Hence, $g(\ES_p-\E)$ is a transformed deviation or variability  measure, and a monetary risk measure satisfying [$p$-CA] can be seen as a mean-deviation functional.


We continue to  characterize the classes of  convex,\footnote {A convex risk measure is a monetary risk measure which also satisfies \emph{convexity}: $\rho(\lambda X+(1-\lambda)Y)\le \lambda\rho(X)+(1-\lambda)\rho(Y)$ for all $\lambda\in[0,1]$.} coherent,\footnote{A coherent risk measure is a convex risk measure which  also satisfies   {\em positive homogeneity}: $\rho(\lambda X)=\lambda \rho(X)$ for all $\lambda\in(0,\infty)$ and $X\in\X$.}  and  consistent risk measures that satisfy [$p$-CA].  These three classes of risk measures  are all monetary risk measure,  and thus they can  be  represented as  the form  in Theorem \ref{th-p}.
Note that for $\rho$ satisfying [$p$-CA], there is a one-to-one correspondence between $\rho$ and $g$ in \eqref{eq-p}, and hence the above classes can be identified based on properties of $g$.
The gap between convex risk measure and consistent risk measure is established clearly in  the following  proposition. In particular, convexity of   $g$ is equivalent to convexity  of $\rho$.


\begin{proposition}\label{prop:nonconvex}
Let $p\in (0,1)$ and $\rho$ be a risk measure on $L^\infty$ satisfying {\rm[}$p$-{\rm CA}{\rm ]} and $\rho(0)=0$.
\begin{enumerate}[(i)]
\item  $\rho$ is a consistent risk measure  if and only if $\rho(X)=g(\ES_p(X)-\E[X])+\E[X]$ for some increasing  and  $1$-Lipschitz  function $g: [0,\infty)\to\R$ with $g(0)=0$.
\item
 $\rho$ is a convex risk measure if and only if $\rho(X)=g(\ES_p(X)-\E[X])+\E[X]$ for some increasing, convex and  $1$-Lipschitz  function $g: [0,\infty)\to\R$ with $g(0)=0$.

\item  $\rho$ is a coherent risk measure if and only if $\rho(X)=\alpha \ES_p(X)+(1-\alpha)\E[X]$ for some $\alpha\in[0,1]$.
\end{enumerate}
\end{proposition}

\begin{proof}
(i) is  implied by Theorem \ref{th-p}.
To see (ii), applying Theorem \ref{th-p}, it is sufficient to prove that
convexity   of the function $g$ in \eqref{eq-p} is equivalent to convexity  of $\rho$.  Note that $g$ is an increasing function.  If  $g$ is convex, then  $\rho$ is a convex risk measure because expectation is linear and $\ES_p$ is a convex risk measure. If  $g$ is nonconvex, then there exist $0\le x<y$ and $\lambda\in(0,1)$ such that $g(\lambda x+(1-\lambda)y)>\lambda g(x)+(1-\lambda)g(y)$. Suppose ($X,Y$) is $p$-concentrated, and satisfies $\ES_p(X)-\E[X]=x$ and $\ES_p(Y)-\E[Y]=y$. Thus, we have
\begin{align*}
\rho(\lambda X+(1-\lambda)Y)&=g(\ES_p(\lambda X+(1-\lambda)Y)-\E[\lambda X+(1-\lambda)Y])+\E[\lambda X+(1-\lambda)Y]\\
&=g(\lambda(\ES_p(X)-\E[X])+(1-\lambda)(\ES_p(Y)-\E[Y]))+\lambda\E[X]+(1-\lambda)\E[Y]\\
&=g(\lambda x+(1-\lambda)y)+\lambda\E[X]+(1-\lambda)\E[Y]\\
&>\lambda g(x)+(1-\lambda)g(y)+\lambda\E[X]+(1-\lambda)\E[Y]\\
&=\lambda (g(\ES_p(X)-\E[X])+\E[X])+(1-\lambda)(g(\ES_p(Y)-\E[Y])+\E[Y])\\
&=\lambda\rho(X)+(1-\lambda)\rho(Y),
\end{align*}
which implies $\rho$ is nonconvex.
(iii) Sufficiency is straightforward. To show necessity,    let  $X$ be such that  $\E[X]=0$ and $\ES_p(X)=x>0$. By Theorem \ref{th-p}, coherence of $\rho$ implies that for all $\lambda>0$,
$$
g(\lambda x)=\rho(\lambda X)=\lambda\rho(X)=\lambda g(x).
$$
This means that  $g$ is  linear  on   $(0,\infty)$. Noting that $g$ is 1-Lipschitz  with $g(0)=0$,  we have $g(x)=\alpha x$ for some $\alpha\in[0,1]$. Hence, we complete the proof of (iii).
\end{proof}

 Since SSD-consistency is strictly weaker than convexity for a law-invariant risk measure, the class of consistent risk measures generalizes that of law-invariant convex risk measures. However, explicit formulas for  nonconvex consistent risk measures are rare in the literature; indeed, all examples in \cite{MW20} involve taking an infimum over convex risk measures. Proposition \ref{prop:nonconvex} leads to many examples of  consistent risk measures with explicit formulas which are outside the classic framework of convex risk measures.

\subsection{A new characterization of the Expected Shortfall}
 Next, we add lower semicontinuity  [{\rm P}] to the requirements in Theorem \ref{th-p} and obtain a new characterization of ES. Remarkably, although Theorem \ref{th-p} allows for many choices of risk measures satisfying [$p$-CA],  lower semicontinuity is enough to  force the function  $g$ in \eqref{eq-p}  to collapse to the identity. Hence, for this characterization of ES, we do not need to assume coherence or convexity.
\begin{theorem}\label{th-ES}
Let $p\in (0,1)$ and $\rho$ be a risk measure on $L^\infty$.
Then $\rho$ satisfies {\rm [M], [TI], [P], [$p$-CA]} and $\rho(0)=0$ if and only if it is ${\rm ES}_p$.
\end{theorem}

\begin{proof}
Sufficiency follows from Proposition 1 and Theorem 5 of \cite{WZ21}. To see necessity, we first apply the result in Theorem \ref{th-p-CA} that $\rho$ has the form $f({\rm ES}_p,\E)$, and the function $y\mapsto f(x,y)$   is increasing on $(-\infty,x]$ for all $x\in\R$.
Next, we will verify that the value of $f$ is independent of its second argument.  On the one hand,  we have $f(x,x)\ge f(x,y)$ for all $x\ge y$. On the other hand, define a sequence of random variables $\{X_n\}_{n\in\N}$ such that $\p(X_n=x)=1-1/n$, $\p(X_n=x-n(x-y))=1/n$ and $X_n\to x$ a.s.. By the property [P], we have
$$
f(x,y)=\liminf_{n\to\infty} f(\ES_p(X_n),\E[X_n])\ge f(x,x).
$$
Therefore, we conclude that $f(x,x)=f(x,y)$ for all $x\ge y$, and this means $\rho(X)=g(\ES_p(X))$ for some function $g$. Finally, using [TI] and $\rho(0)=0$, one can  conclude that $g$ is the identity.\end{proof}


We can equivalently  express Theorem \ref{th-ES}   in terms of the acceptance set as in the next proposition. A proof is straightforward from the definition of an acceptance set.

\begin{proposition}\label{th-acceptance}
Let $p\in(0,1)$.
An acceptance set $\mathcal A$ satisfies \begin{enumerate}[(i)]
\item $(X, Y)$ is  $p$-concentrated  and
$X+Y\in\mathcal A$ $\Longrightarrow$ $X'+Y'\in\mathcal A$   for all $X', Y'$  with $X'\laweq X$, $Y'\laweq Y$,
\item $X_n\in\mathcal A$ for each $n=1,2,\dots$ and $X_n\to X$ pointwise $\Longrightarrow$ $X\in\mathcal A$, and
\item $\sup\{c\in\R: c\in\mathcal A\}=0$,
\end{enumerate}
 if and only if  $\mathcal A$ is the acceptance set of $\ES_p$.
\end{proposition}

\section{Generalization to larger spaces}\label{sec:5}
In this section, we  generalize the  characterization results in Sections \ref{sec:3} and \ref{sec:4} to larger $L^q$ spaces than $L^\infty$. The risk measure $\rho:L^q\to \R$ will be assumed to take real values.

\subsection{Generalization to $L^q$ for $q\ge 1$}
We endow the natural norm on $L^q$, $q\in[ 1,\infty)$, i.e., $\|X\|_q=(\E[|X|^q])^{1/q}$ for $X\in L^q$, and   continuity is defined with respect to $\|\cdot\|_q$. Furthermore, we recall the notation $\H$ as the half-space $\{(x,y)\in\R^2: x\ge y\}$.


\begin{proposition}\label{th-large-pCA}
Let $p\in (0,1)$, $q\ge 1$ and
$\rho:L^q\to\R$ be a continuous risk measure. Then,
\begin{itemize}
\item[(i)]
$\rho$ satisfies {\rm [$p$-CA]} if and only if it has the form
$f({\rm{ES}}_p,\E)$, where $f:\H\to \R$ is a continuous bivariate function which is increasing in its first argument.
\item[(ii)]
$\rho$ satisfies {\rm [M]} and {\rm [$p$-CA]} if and only if it has the form
$f({\rm{ES}}_p,\E)$, where $f:\H\to \R$ is a continuous bivariate function which is increasing in both arguments.
\end{itemize}
\end{proposition}

\begin{proof}
Sufficiency  in both    (i) and (ii) is trivial. To see necessity,
noting that for any $X\in L^q$, there exists a sequence $\{X_n\}_{n\in\N}\subseteq L^\infty$ converges to $X$ with respect to the norm $\|\cdot\|_q$. By the continuity of $\rho$,   the statements in Theorem \ref{th-p-CA} are all valid. Thus, it remains to prove that $f$ is continuous on $\H$. For $(x_0,y_0)\in\H$, let $\{(x_n,y_n)\}_{n\in\N}\subseteq \H$ be a sequence converges to $(x_0,y_0)$. Let $A\in\mathcal F$ such that $\p(A)=p$. Define
a sequence of random variables
$$
X_n(\omega)=\frac{y_n-(1-p)x_n}{p}~{\rm for~}\omega\in A,~~~\mbox{and}~~~X_n(\omega)=x_n~{\rm for~}\omega\in A^c,
$$
 and let $$
X(\omega)=\frac{y_0-(1-p)x_0}{p}~{\rm for~}\omega\in A,~~~\mbox{and}~~~X(\omega)=x_0~{\rm for~}\omega\in A^c.
$$
Obviously, $\ES_p(X_n)=x_n$, $\E[X_n]=y_n$ and $X_n\to X$ in $L^q$ with $\ES_p(X)=x_0$, $\E[X]=y_0$.
Hence, we have
$$
f(x_n,y_n)=f({\rm{ES}}_p(X_n),\E[X_n])=\rho(X_n)\to\rho(X)=f({\rm{ES}}_p(X),\E[X])=f(x_0,y_0).
$$
This completes the proof.
\end{proof}


Similarly, Theorems \ref{th-p} and \ref{th-ES} can be generalized to $L^q$ for $q\ge 1$.
\begin{proposition}\label{th-large-p}
Let $p\in (0,1)$, $q\ge 1$  and
$\rho:L^q\to\R$ be a continuous   risk measure. Then $\rho$
satisfies {\rm [M]}, {\rm [TI]}, {\rm [$p$-CA]} and $\rho(0)=0$ if and only if it has the form  $
\rho(X)=g(\ES_p(X)-\E[X])+\E[X]
 $
for some  increasing and $1$-Lipschitz function $g: [0,\infty)\to\R$  with $g(0)=0$.
In particular, such $\rho$ is a consistent risk measure.
\end{proposition}

\begin{proposition}\label{th-large}
Let $p\in (0,1)$, $q\ge 1$  and
$\rho:L^q\to\R$ be a  continuous risk measure. Then $\rho$ satisfies {\rm [M], [TI], [P], [$p$-CA]} and $\rho(0)=0$ if and only if it is ${\rm ES}_p$.
\end{proposition}

\subsection{Impossibility results on $L^q$ for $q\in [0,1)$}\label{subsec:5-2}
 In this section,  we let $q\in [0,1)$ and   consider the larger spaces $L^q\supset L^1$   as the domain of the risk measure $\rho$. It is shown in Theorem 2 of \cite{WZ21} that the only real-valued risk measure on $L^q$ satisfying   [M], [LI], [P] and  [NRC] is the constant risk measure $\rho=0$.
 A natural question arises: Is there a nonconstant risk measure $\rho: L^q\to\R$  satisfying [$p$-CA]?
We shall first see in the following example that [$p$-CA] on $L^q$ does not necessarily lead to a constant risk measure.
\begin{example}\label{ex-large}
Let $f(x,y)$ be a  bounded real function on $\H=\{(x,y)\in\R^2: x\ge y\}$ which is increasing  in both $x,y$ and $M>0$ be such that $|f|\le   M$. Define
\begin{align*}
\rho(X)=\begin{cases}
f(\ES_p(X),\E[X]),~~~&X\in L^1,\\
-M,~~~&\E[X_-]=\infty,~\E[X_+]<\infty,\\
M,~~~&\E[X_+]=\infty,
\end{cases}
\end{align*}
where $X_{+}=\max \{X, 0\}$ and $X_{-}=\max \{-X, 0\}$. One can verify that $\rho$ satisfies [M] and [$p$-CA].
\end{example}

As illustrated by Example \ref{ex-large},
in contrast to [NRC], we can construct a class of nontrivial risk measures  bounded on $L^0$
that satisfies [$p$-CA].
Nevertheless,
the following proposition illustrates that it is pointless to consider monotone risk measures $\rho: L^q\to\R$  satisfying [$p$-CA] if   $\rho$ is unbounded on the set of constants. As a consequence, we conclude that the domain $L^1$ is the most natural, and essentially the largest, choice for any real-valued risk measures  satisfying  {\rm [M], [TI] and [$p$-CA]}.

\begin{proposition}\label{prop-large}
Let $p\in(0,1)$ and $q\in [0, 1)$. There is no such $\rho: L^q\to\R$  that satisfies {\rm [M], [$p$-CA]} and $\lim_{c\to\infty}\rho(c)=\infty$.
\end{proposition}

\begin{proof}
Assume that such $\rho$ exists.
Take a nonnegative $X\in L^q\setminus L^1$, and let $X_n=\min\{X,n\}\in L^\infty$  for $n\in\N$. Obviously, we have $X_n\uparrow X$. By Theorem \ref{th-p-CA}, $\rho$ has the form $f(\ES_p,\E)$ on $L^\infty$. It then follows from  [M] and
the condition $\lim_{c\to\infty}\rho(c)=\infty$ that $\lim_{y\to\infty}f(x,y)=\infty$.
Note that $\E[X_n]\to \infty$. Thus, we obtain
$$
\rho(X)\ge \liminf \rho(X_n)=\liminf f(\ES_p(X_n),\E[X_n])=\infty,
$$ a contradiction.
\end{proof}

Since a monetary risk measure $\rho$ necessarily satisfies $\lim_{c\to\infty} \rho(c)=\infty$, we conclude from Proposition \ref{prop-large} that  for $q\in [0, 1)$,   there is no monetary risk measure $\rho: L^q\to\R$ that satisfies    {\rm [$p$-CA]}.



\section{An economic reasoning  for concentration aversion}
\label{sec:6-r1}
For the key concept of concentration aversion in this paper,
it is assumed  in Definition \ref{def:CA} that there exists a tail event $A$ of regulatory concern.
Such a tail event $A$ is exogenous to the property [CA]; similarly, the structure of $p$-concentration is exogenous to the property [$p$-CA].
For a solid economic foundation of using [CA],
it would be more compelling to
justify the structure of $p$-concentration from   endogenous reasoning.\footnote{We thank an anonymous referee for bringing  this question up.}
Addressing this issue is the objective of this section.
We will show that, if a regulator is concerned about dangerous dependence structures satisfying a few axioms, then [CA] must hold for the regulator's risk measure.

Assume   $\X =L^\infty$ in this section, and denote  by $L^\infty_c\subseteq L^\infty  $ the set of all continuously distributed random variables in $L^\infty$. We focus on continuous distributions because we will work with dependence  structures, which will be modelled by copulas. {An $n$-copula is a  joint distribution function on $\R^n$ with standard uniform marginals. Sklar's theorem implies that the joint distribution $F$ of any random vector $\mathbf X$ can be expressed by a copula $C$ of $\mathbf X$ through $F(x_1,\dots,x_n)=C(F_1(x_1),\dots,F_n(x_n))$ where $F_1,\dots,F_n$ are the marginals of $F$. The copula $C$ is unique if $F_1,\dots,F_n$ are continuous.} We denote by $C_{\mathbf X}$ the copula of $\mathbf X$ if it is unique, and $\mathcal C_n$ the set of $n$-copulas. We refer to \cite{J14} for a general treatment of copulas.

Suppose that a regulator is concerned about random losses that are dependent in an adverse (dangerous) way.
The interpretation of dangerousness of a dependence structure is modelled by a set $\mathcal D\subseteq \mathcal C_2$. We will specify a suitable $\mathcal D$ later, but a primary example is
  \begin{align}
  \label{eq:Dpalt}
 \mathcal D_p  =\{C\in \mathcal C_2: C(p,p)=p\},~~~~p\in (0,1).
 \end{align}
By Theorem 3 of \cite{WZ21}, a copula of $(X,Y)$ is in $\mathcal D_p$ if and only if $(X,Y)$ is $p$-concentrated; hence, $\mathcal D_p$ is the set of bivariate copulas for $p$-concentrated random vectors.
 Since an adverse dependence structure bears more risk, the regulatory risk measure $\rho: L^\infty \to \R$
should
 satisfy \emph{$\mathcal D$-aversion}, that is, $\rho(X+Y)\le \rho(Z+W)$ for all  $X,Y,Z,W\in L^\infty$ satisfying $X\laweq Z$, $Y\laweq W$, and  a copula of $(Z,W)$ is in $\mathcal D$.
For the special case of $\mathcal D=\mathcal D_p$ for some $p\in (0,1)$,    $\mathcal D_p$-aversion is precisely [$p$-CA].\footnote{If we insist using copulas for continuously distributed random variables, we may alternatively require $\rho(X+Y)\le \rho(Z+W)$ to hold only for $X,Y,Z,W\in L^\infty_c$ and  $(Z,W)$ with a unique copula in $\mathcal D_p$. This property  is slightly weaker than [$p$-CA], but they are equivalent if $\rho$ is monotone and lower semicontinuous with respect to a.s.~convergence.}

In what follows, we discuss  reasonable choices of $\mathcal D$ for the regulator.
A common idea of diversification originates from the Law of Large Numbers (LLN), or its   refined versions, the Central Limit Theorems.
A dependence structure of risks is arguably quite dangerous if there is no effect of LLN; that is, the average risk does not vanish even if the number of risks in the pool tends to infinity. Inspired by this observation, we define non-diversifiability via violation of LLN.
For a  copula $C\in\mathcal C_2$, we say that a sequence $(X_n)_{n\in\N}$
is \emph{sequentially $C$-coupled} if  $X_n\laweq X_{n+1}$ and  $ C $ is the copula of $(X_n,X_{n+1})$ for each $n\in \N$. Note that the dependence of $(X_n,X_k)$ for $|n-k|>1$ is unspecified and it typically has some flexibility.
We say that $C\in\mathcal C_2$ is  \emph{non-diversifiable} if each sequentially $C$-coupled sequence $(X_n)_{n\in\N}$ in $L^\infty_c$
 breaks LLN, that is,
\begin{align}
\label{eq:breakthelaw}
\frac 1n \sum_{i=1}^n X_i -\mu \not \pcn 0~~~\mbox{as $n\to \infty$,}
\end{align}
where $\mu$ is the mean of $X_1$; otherwise $C$ is \emph{diversifiable}.
A simple example of a non-diversifiable copula is the comonotonic copula $C^+$, defined via $C^+(u,v)=\min(u,v)$.
For any sequentially $C^+$-coupled sequence $(X_n)_{n\in\N} $ in $L^\infty_c$, due to comotonicity, we have  $X_n=X_{n+1}$ for $n\in \N$, and
\begin{align}
\label{eq:breakthelaw2}
\frac 1n \sum_{i=1}^n X_i -\mu  = X_1 -\mu \not \pcn 0~~~\mbox{as $n\to \infty$}.
\end{align}   On the other hand, the independent copula $C^\perp$, defined via $C^\perp(u,v)=uv$,  is clearly diversifiable due to LLN.

Another important consideration is that positive dependence is more dangerous than negative dependence.
Recall that for two bivariate copulas $C$ and $C'$, the point-wise order $C'\ge C$, called the concordance order (see e.g., \cite{MS02}), compares the level of positive dependence.
 In particular, if $X\laweq Z$ and $Y\laweq W$ and $C_{X,Y}\ge C_{Z,W}$, then  $Z+W\preceq_{\rm SSD}X+Y$ (see e.g., \cite{WW20}), and thus $X+Y$ bears more risk than $Z+W$ in a commonly agreed sense of riskiness. The comonotonic  copula $C^+$ attains the maximum in concordance order.

Finally, a combination of dangerous scenarios, in the form  of a probability mixture, is still dangerous,
 because such a mixture  represents randomly picking a dangerous scenario.

Translating the above considerations into properties of $\mathcal D$, we define
a \emph{bivariate concentration class} which is a subset $\mathcal D$ of $\mathcal C_2$ satisfying the following three properties [ND], [DM] and [Cx]. \begin{enumerate}
\item[{[ND]}] Non-diversifiability: Each $C\in \mathcal D$ is non-diversifiable.
\item[{[DM]}] Dependence monotonicity:  If $C\in \mathcal D$ and $  C\le C'\in \mathcal C_2$, then $C'\in\mathcal D$.
\item[{[Cx]}]  Convexity: If $C_n\in \mathcal D$ for $n\in \N$, then $\sum_{n\in \N} \lambda_n C_n \in \mathcal D$ for any non-negative numbers $\lambda_n$, $n\in \N$ with $ \sum_{n\in \N} \lambda_n =1$.
\end{enumerate}
The first property,  [ND], simply means that
each dependence structure  in $\mathcal D$   breaks LLN.
The second property, [DM], says that if $C$ is considered dangerous and   $C'$ is more positively dependent than $C$, then $C'$ is also considered as dangerous. Convexity [Cx] means combining dangerous scenarios leads to a dangerous scenario.
The three properties are arguably  quite natural for a concept of concentration of interest to a regulator.

We first verify a few important examples of bivariate concentration classes.
\begin{proposition}
\label{prop:r1-1}
The sets
$\mathcal D_p$ for $p\in (0,1)$  and the singleton $\{C^+\}$
are bivariate concentration classes.
\end{proposition}
\begin{proof}
We first verify the statement for $\{C^+\}$.
A singleton is obviously convex, and thus [Cx] holds.
By \eqref{eq:breakthelaw2},  $\{C^+\}$
 satisfies [ND]. Since $C^+$ is the maximum in concordance order, $\{C^+\}$ satisfies [DM].
 Next, we show that $\mathcal D_p$ for $p\in (0,1)$ satisfies [Cx], [ND] and [DM].
 The property [Cx]   follows directly from \eqref{eq:Dpalt}.
 Note that $C'(p,p)\le C^+(p,p)=p$ for all $C'\in \mathcal C_2$.
 Using \eqref{eq:Dpalt}, $C'\ge C\in \mathcal D_p$ implies $C'(p,p)=C(p,p)=p$, and thus $C'\in \mathcal D_p$. This shows that $\mathcal D_p$ satisfies [DM].
To show [ND], take $C\in \mathcal D_p$ and
construct any sequentially $C$-coupled sequence $(X_n)_{n\in \N}$.
Since $(X_n,X_{n+1})$ is $p$-concentrated, by Corollary A.1 of \cite{WZ21}, $X_n$ and $X_{n+1}$ share the same a.s.~unique $p$-tail event $A=\{X_n>x_p\}=\{X_{n+1}>x_p\}$ where $x_p=\VaR_p(X_1)$. Applying this argument to $n\in \N$, we know that $X_1,X_2,\dots$ share the same tail event $A$ which does not depend on $n$. Write $\widetilde X_n= n^{-1}  \sum_{i=1}^n X_i$.
We note that $\widetilde X_n\id_A$ does not converge to $\E[X_1] \id_A$ since
$$
\E\left[ \widetilde X_n \mid A\right] =\E[X_1\mid A] =\ES_p(X_1) > \E[X_1]
$$
and $(\widetilde X_n)_{n\in \N}$ is uniformly integrable.
Therefore,
\begin{align*}
\widetilde X_n  =  \widetilde X_n \id_{A} + \widetilde X_n \id_{A^c}  \not \pcn \E[X_1],
\end{align*}
and  thus $(X_n)_{n\in \N}$ is non-diversifiable.
This shows that $\mathcal D_p$ satisfies [ND]. Therefore, $\mathcal D_p$ is a bivariate concentration class.
\end{proof}
 If both $\mathcal D$ and $\mathcal D'$
 are bivariate concentration classes, then so is  $\mathcal D\cap \mathcal D'$.
 Using this relation, we can construct   bivariate concentration classes other than the ones in Proposition \ref{prop:r1-1}.
  The next result, which is the main technical result in this section, shows that a bivariate concentration class may not contain anything more than those in $ \mathcal D_p$ for some $p\in(0,1)$.

 \begin{theorem}\label{th:r1-1}
A  copula $C\in \mathcal C_2$
  is in  a bivariate concentration class
 if and only if
 $C\in \mathcal D_p$ for some $p\in (0,1)$.
\end{theorem}
\begin{proof}
The ``if" statement follows directly from Proposition \ref{prop:r1-1}. Below we will show the ``only if" statement.
We first present some technical preparations.
For any copula $C$ and $p\in (0,1)$, denote by $t_C(p)$  the essential infimum of
the distribution of $V$ given $U\ge p$ where $(U,V)\sim C$,
and by  $s_C(p)$   the essential supremum of
the distribution of $V$ given $U \le  p$.
 Since $\p(U\ge p)=1-p$, the essential infimum of
the distribution of $V$ given $U\ge p$ is at most $p$, which is the case when $\{U\ge p\}$ is a $p$-tail event of $V$.
Therefore, $t_C(p)\le p$ and similarly, $s_C(p)\ge p$ for $p\in (0,1)$.
Moreover, both $t_C$ and $s_C$ are increasing curves on $(0,1)$. See Figure \ref{fig:r1-1}a for an illustration, where the grey area is  between the curves $t_C$ and $s_C$.

To proceed, we need the following lemma, which may be of   interest  in dependence theory by its own right. 

\begin{lemma}\label{lem:r1-1}
For any $C\in \mathcal C_2$, there exists a copula  $\widetilde C\ge C$ such that
$\widetilde C $ has positive density (possibly plus a non-density component) on the region $B_C:=\{(u,v)  \in [0,1]^2: t_C(u)< v < s_C(u)\}$.
\end{lemma}

The proof of Lemma \ref{lem:r1-1} requires   some delicate constructions of copulas, and it is put in Appendix \ref{app:A}. Here, we briefly explain the intuition behind the proof. For any given copula $C$ supported in a subset of $B_C$ possibly with no density (see Figure \ref{fig:r1-1}a), we first mix it with the comonotonic copula $C^+$, so that the resulting copula $C' \ge C$ has a support that includes the diagonal line in $[0,1]$ (see Figure \ref{fig:r1-1}b). Second, we run a continuum of  concordance-increasing (CI) transfers of \cite{T80} (see Figure \ref{fig:r1-1}b-c) on 
$C'$ to obtain another copula $\widehat C \ge C'$ which has positive density on a subset of $B_C$. 
Finally, we run another continuum of   CI transfers on $\widehat C$ to arrive at a copula $\widetilde C \ge \widehat C$ which has positive density on $B_C$  (see Figure \ref{fig:r1-1}c-d).  

\begin{figure}[t]    
 \begin{subfigure}[b]{0.23\textwidth}
         \centering
\begin{tikzpicture}
\draw[<->] (0,3.4) -- (0,0) -- (3.4,0);
\draw[gray,dotted] (0,0.6) -- (3,0.6);
\draw[gray,dotted] (0,1.2) -- (3,1.2);
\draw[gray,dotted] (0,1.8) -- (3,1.8);
\draw[gray,dotted] (0,2.4) -- (3,2.4);
\draw[gray,dotted] (0.6,0) -- (0.6,3);
\draw[gray,dotted] (1.2,0) -- (1.2,3);
\draw[gray,dotted] (1.8,0) -- (1.8,3);
\draw[gray,dotted] (2.4,0) -- (2.4,3);
\draw[gray,dotted] (3,0) -- (3,3);
\draw[gray,dotted] (0,3) -- (3,3); 
\node[below] at (3,0) {$1$}; 
\node[left] at (0,3) {$1$}; 
\node[below] at (0,0) {$0$}; 
\draw[thick] (0,0.6) -- (0.6,0);
\draw[thick] (1.2,1.2) -- (1.8,1.8);
\draw[thick] (0.6,1.8) -- (1.2,2.4); 
\draw[thick] (1.8,0.6) -- (2.4,1.2);
\draw[thick] (2.4,2.4) -- (3,3);  
\fill [top color=gray, bottom color=gray, opacity=0.05] (0,0)--(0,0.6)--(0.6,0.6)--(0.6,1.8)--(1.2,2.4)--(2.4,2.4)--(2.4,1.2)--(1.8,0.6)--(0.6,0.6)--(0.6,0)--(0,0); 
\end{tikzpicture} 
         \caption{\footnotesize Black lines: support of $C$; grey area: $B_C$}
         \label{fig:r1-2-1}
     \end{subfigure}
     ~
      \begin{subfigure}[b]{0.23\textwidth}
         \centering
\begin{tikzpicture}
\draw[<->] (0,3.4) -- (0,0) -- (3.4,0);
\draw[gray,dotted] (0,0.6) -- (3,0.6);
\draw[gray,dotted] (0,1.2) -- (3,1.2);
\draw[gray,dotted] (0,1.8) -- (3,1.8);
\draw[gray,dotted] (0,2.4) -- (3,2.4);
\draw[gray,dotted] (0.6,0) -- (0.6,3);
\draw[gray,dotted] (1.2,0) -- (1.2,3);
\draw[gray,dotted] (1.8,0) -- (1.8,3);
\draw[gray,dotted] (2.4,0) -- (2.4,3);
\draw[gray,dotted] (3,0) -- (3,3);
\draw[gray,dotted] (0,3) -- (3,3); 
\node[below] at (3,0) {$1$}; 
\node[left] at (0,3) {$1$}; 
\node[below] at (0,0) {$0$}; 
\draw[thick] (0,0.6) -- (0.6,0);
\draw[thick] (0,0) -- (3,3);
\draw[thick] (0.6,1.8) -- (1.2,2.4); 
\draw[thick] (1.8,0.6) -- (2.4,1.2);
\draw[thick] (2.4,2.4) -- (3,3);  
\fill [top color=gray, bottom color=gray, opacity=0.05] (0,0)--(0,0.6)--(0.6,0.6)--(0.6,1.8)--(1.2,2.4)--(2.4,2.4)--(2.4,1.2)--(1.8,0.6)--(0.6,0.6)--(0.6,0)--(0,0); 
\draw [->, blue, dashed]  (1.5,1.5)--(1.5,2.05); 
\draw [->, blue, dashed]  (1.5,1.5)--(0.95,1.5); 
\draw [->, blue, dashed]  (0.9,2.1)--(1.45,2.1); 
\draw [->, blue, dashed]  (0.9,2.1)--(0.9,1.55);
\filldraw[black] (1.5,1.5) circle (1pt);  
\filldraw[blue] (1.5,2.1) circle (1pt);  
\filldraw[black] (0.9,2.1) circle (1pt);  
\filldraw[blue] (0.9,1.5) circle (1pt);  
\end{tikzpicture} 
         \caption{\footnotesize Black lines: support of $C'$; arrows: a CI transfer}
         \label{fig:r1-2-2}
     \end{subfigure}    
     ~
        \begin{subfigure}[b]{0.23\textwidth}
         \centering
\begin{tikzpicture}
\draw[<->] (0,3.4) -- (0,0) -- (3.4,0);
\draw[gray,dotted] (0,0.6) -- (3,0.6);
\draw[gray,dotted] (0,1.2) -- (3,1.2);
\draw[gray,dotted] (0,1.8) -- (3,1.8);
\draw[gray,dotted] (0,2.4) -- (3,2.4);
\draw[gray,dotted] (0.6,0) -- (0.6,3);
\draw[gray,dotted] (1.2,0) -- (1.2,3);
\draw[gray,dotted] (1.8,0) -- (1.8,3);
\draw[gray,dotted] (2.4,0) -- (2.4,3);
\draw[gray,dotted] (3,0) -- (3,3);
\draw[gray,dotted] (0,3) -- (3,3); 
\node[below] at (3,0) {$1$}; 
\node[left] at (0,3) {$1$}; 
\node[below] at (0,0) {$0$}; 
\draw[thick] (1.2,1.2) -- (1.8,1.8);  
\draw[thick] (2.4,2.4) -- (3,3);  
\fill [top color=gray, bottom color=gray, opacity=0.05] (0,0)--(0,0.6)--(0.6,0.6)--(0.6,1.8)--(1.2,2.4)--(2.4,2.4)--(2.4,1.2)--(1.8,0.6)--(0.6,0.6)--(0.6,0)--(0,0); 
 \fill [top color=blue, bottom color=blue, opacity=0.1] (0,0)--(0,0.6)--(0.6,0.6)--(0.6,0)
 --(0,0);
 \fill [top color=blue, bottom color=blue, opacity=0.1]  (0.6,1.8)--(1.2,2.4)--(1.2,1.2)-- (2.4,1.2)--(1.8,0.6)--(0.6,0.6);
  \fill [top color=blue, bottom color=blue, opacity=0.1]   (1.2,2.4)--(2.4,2.4)-- (2.4,1.2)--(1.8,1.2)--(1.8,1.8)--(1.2,1.8); 
  \draw [->, blue, dashed]  (1.4,1.4)--(1.4,1.65); 
\draw [->, blue, dashed]  (1.4,1.4)--(0.95,1.4); 
\draw [->, blue, dashed]  (0.9,1.7)--(1.35,1.7); 
\draw [->, blue, dashed]  (0.9,1.7)--(0.9,1.45);
\filldraw[black] (1.4,1.4) circle (1pt);  
\filldraw[blue] (1.4,1.7) circle (1pt);  
\filldraw[black] (0.9,1.7) circle (1pt);  
\filldraw[blue] (0.9,1.4) circle (1pt);  
\end{tikzpicture} 
         \caption{\footnotesize Blue area: density area of $\widehat C $; arrows: a CI transfer}
         \label{fig:r1-2-3}
     \end{subfigure} ~
             \begin{subfigure}[b]{0.23\textwidth}
         \centering
\begin{tikzpicture}
\draw[<->] (0,3.4) -- (0,0) -- (3.4,0);
\draw[gray,dotted] (0,0.6) -- (3,0.6);
\draw[gray,dotted] (0,1.2) -- (3,1.2);
\draw[gray,dotted] (0,1.8) -- (3,1.8);
\draw[gray,dotted] (0,2.4) -- (3,2.4);
\draw[gray,dotted] (0.6,0) -- (0.6,3);
\draw[gray,dotted] (1.2,0) -- (1.2,3);
\draw[gray,dotted] (1.8,0) -- (1.8,3);
\draw[gray,dotted] (2.4,0) -- (2.4,3);
\draw[gray,dotted] (3,0) -- (3,3);
\draw[gray,dotted] (0,3) -- (3,3); 
\node[below] at (3,0) {$1$}; 
\node[left] at (0,3) {$1$}; 
\node[below] at (0,0) {$0$};  
\draw[thick] (2.4,2.4) -- (3,3);  
\fill [top color=gray, bottom color=gray, opacity=0.05] (0,0)--(0,0.6)--(0.6,0.6)--(0.6,1.8)--(1.2,2.4)--(2.4,2.4)--(2.4,1.2)--(1.8,0.6)--(0.6,0.6)--(0.6,0)--(0,0); 
 \fill [top color=blue, bottom color=blue, opacity=0.1] (0,0)--(0,0.6)--(0.6,0.6)--(0.6,0)
 --(0,0);
 \fill [top color=blue, bottom color=blue, opacity=0.1]  (0.6,1.8)--(1.2,2.4)--(1.2,1.2)-- (2.4,1.2)--(1.8,0.6)--(0.6,0.6);
  \fill [top color=blue, bottom color=blue, opacity=0.1]   (1.2,2.4)--(2.4,2.4)-- (2.4,1.2)--(1.8,1.2)--(1.8,1.8)--(1.2,1.8);
    \fill [top color=blue, bottom color=blue, opacity=0.1]   (1.2,1.2)--(1.2,1.8)-- (1.8,1.8)--(1.8,1.2); 
\end{tikzpicture} 
         \caption{\footnotesize Blue area: density area of $\widetilde C $, which equals $B_C$}
         \label{fig:r1-2-4}
     \end{subfigure}

\caption{Intuition behind Lemma \ref{lem:r1-1}; here, $C$ is a $5\times 5$ checkerboard copula} 
          \label{fig:r1-1}
\end{figure}
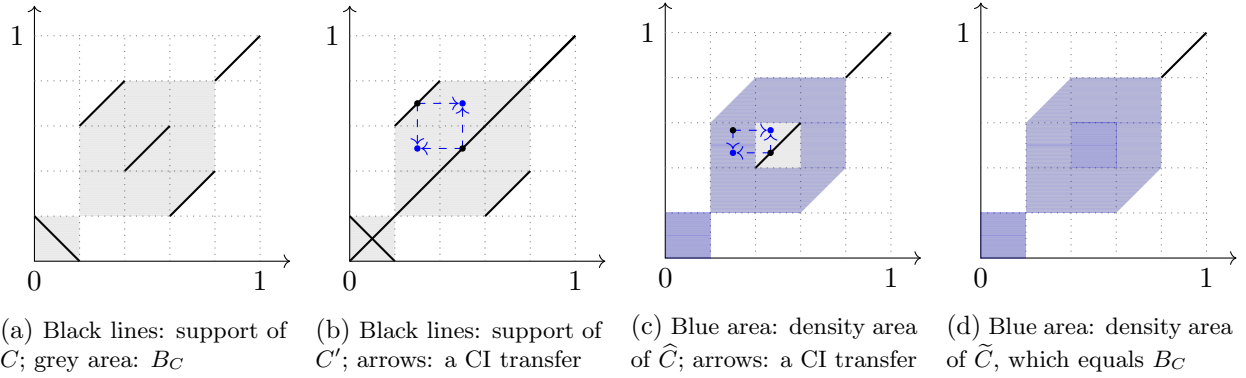

We   continue to prove Theorem \ref{th:r1-1}.
Let $\mathcal D$ be a bivariate concentration class and take $C\in \mathcal D$. Suppose for the purpose of contradiction that there does not exist $p\in (0,1)$ such that $C(p,p)=p$.
Note that if $s_C(p)=p$, then $C(p,p)=\p(U\le p,V\le p)= \p(U\le p) =p$. Similarly, if $t_C(p)=p$, then $C(p,p)=p$. Hence, our assumption on $C$ implies that $t_C(p)<p<s_C(p)$ for all $p\in (0,1)$.

Take $\widetilde C $ as the one in Lemma \ref{lem:r1-1}.
Since $\mathcal D$ is a bivariate concentration class and $\widetilde C \ge C$, we have
$\widetilde C \in \mathcal D$.
Let
\begin{equation}\label{eq:conditionalcopula} \widetilde C _{2|1}(v|u)= \frac{\partial  \widetilde C}{\partial u}(u,v),~~~~\mbox{if the partial derivative exists}.\end{equation}
It is known that $\widetilde C _{2|1}$ is a conditional distribution of $V$ given $U$, where $(U,V)\sim \widetilde C $; see \citet[Section 2.1.3]{J14}. As a consequence, $\widetilde C _{2|1}$ exists almost everywhere and takes value in $[0,1] $.
Let
 $ \widetilde C ^{-1}_{2|1}(v|u)$ be the corresponding conditional $v$-quantile of $V$ given $U=u$ for $u,v\in (0,1)$; we omit ``almost everywhere".

Take  a sequence $(U_n)_{n\in \N}$ of iid  random variables uniformly distributed on $[0,1]$.
We will construct a Markov process $(X_n)_{n\in\N}$ as follows. Let
\begin{align}
\label{eq:Markov}
 X_1=U_1 \mbox{~~~and~~~~}X_{n+1}= \widetilde C ^{-1}_{2|1} (U_{n+1}|X_{n}) \mbox{~for $n\ge 1$}.
\end{align}
By construction, $(X_{n},X_{n+1})$ has the distribution $\widetilde C $ for each $n$; see \citet[Section 6.9]{J14}.
Moreover, $(X_n)_{n\in \N}$ is obviously Markov and stationary.
Since $t_{C}(u)<u<s_C(u)$ for all $u\in (0,1)$
and $ \widetilde C $ has positive density on $B_C$,
we know that
 the Markov process $(X_n)_{n\in \N}$ is irreducible.
Since an irreducible and stationary Markov process
satisfies   LLN  (see Theorem 4.3 of \cite{T96}), we have 
$n^{-1}\sum_{i=1}^n X_n \pcn \E[X_1].$
This shows that $\widetilde C $ is diversifiable, contradicting $\widetilde C \in \mathcal D$.
 Therefore, we conclude that $C(p,p)=p$ for some $p\in (0,1)$, leading to the desired ``only if" statement.
\end{proof}

 \begin{remark}
 We comment on two technical points.
 First, convexity [Cx] is not needed in the proof of the ``only if" direction in Theorem \ref{th:r1-1}.
Therefore, a copula $C$ is in any set satisfying [ND] and [DM] if and only if
 $C\in \mathcal D_p$ for some $p\in (0,1)$.
  Second, to show that $\widetilde C $ is diversifiable in the above proof, we constructed the Markov process $(X_n)_{n\in \N}$ in \eqref{eq:Markov} with uniform marginal distributions.
A strictly increasing transform on $(X_n)_{n\in \N}$ does not matter as  the resulting Markov process is always irreducible.
Hence, we may alternatively define non-diversifiability of a copula by  requiring \eqref{eq:breakthelaw} to  hold only for the sequentially $\widetilde C $-coupled process with Markov dependence \eqref{eq:Markov} and a uniform marginal distribution (or another continuous marginal distribution), and our results in Proposition \ref{prop:r1-1} and Theorem \ref{th:r1-1} remain valid.
 \end{remark}

 Theorem \ref{th:r1-1} leads to the following characterization of $\mathcal D_p$ as an important subclass of bivariate concentration classes.
A   bivariate concentration class $\mathcal D $ is \emph{maximal}  if there does not exist another   bivariate concentration class $\mathcal D' \ne\mathcal D$  containing $\mathcal D$.

  \begin{proposition}\label{prop:r1-2}
The set $\mathcal D\subseteq \mathcal C_2$ is a maximal    bivariate concentration class if and only if  $\mathcal D=  \mathcal D_p$  for some $p\in (0,1)$.
  \end{proposition}
    \begin{proof}
    We first show the ``if" statement.
    To show that $\mathcal D_p$ is maximal, suppose that a   bivariate concentration class $\mathcal D$ satisfies
    $\mathcal D_p\subseteq \mathcal D$  and $\mathcal D$ contains a copula $C\not \in \mathcal D_p$. This means $C(p,p) < p$ since $C(p,p)\le p$ is satisfied by any copula.
    Take another copula $C'\in \mathcal D_p$ satisfying $C'(u,u) < u$ for all $u\in (0,1)\setminus\{p\}$.
    Such a copula can be obtained by, for instance, mixing a Lebesgue measure on $[0,p]^2$ and a Lebesgue measure on $[p,1]^2$. Let $C^*=C/2+C'/2 $ which is in $\mathcal D$ since $\mathcal D$ is convex. We have $C^*(u,u)<u$ for all $u\in (0,1)$, and hence $C^*$ is not in any $\mathcal D_u$.
    By Theorem \ref{th:r1-1}, $C^*$ is not in any bivariate concentration class, a contradiction to $C^*\in\mathcal D$.
    Therefore, $\mathcal D_p$ is maximal.

   Below, we show the ``only if" statement.
   Let $\mathcal D$ be a     bivariate concentration class. If there exists $p\in (0,1)$ such that $C(p,p)=p$ for all $C\in \mathcal D$, then we have $\mathcal D\subseteq \mathcal D_p$, and the maximality of $\mathcal D$ implies $\mathcal D=\mathcal D_p$.

  Next, we suppose that there does not exist $p\in (0,1)$ such that $C(p,p)=p$ for all $C\in \mathcal D$.  We will show that this case is not possible by contradiction.
   Take any $\epsilon\in (0,1/2)$.
   For each $p\in [\epsilon, 1-\epsilon]$, our assumption implies that
   there exists a copula $C^p\in \mathcal D$ such that $C^p(p,p)<p$. Note that $C^p(u,u)$ is continuous in $u$ since all copula functions are Lipchitz continuous. As a consequence, there exists an open interval $I_p$ with $p\in I_p$
   such that $C^p(u,u)<u$ for all $u\in I_p$.
   Clearly, $\{I_p: p\in [\epsilon,1-\epsilon]\}$ is an open cover of $[\epsilon, 1-\epsilon]$.  Since $[\epsilon, 1-\epsilon]$ is compact, there exists a finite subcover, denoted by $\{I_{p_i}: i=1,\dots,n\}$, which satisfies $\bigcup_{i=1}^n I_{p_i}\supseteq [\epsilon,1-\epsilon]$.
 Write $C^{[\epsilon]}=n^{-1} \sum_{i=1}^n C^{p_i}$.
  We have $C^{[\epsilon]}\in \mathcal D$ since $\mathcal D$ is convex. Moreover, $C^{[\epsilon]}$ satisfies $C^{[\epsilon]}(u,u)<u$ for all $u\in [\epsilon,1-\epsilon]$.

Define $ C^*= \sum_{k=1}^\infty 2^{-k} C^{[3^{-k}]}$.
Convexity   of $\mathcal D$ for countable sums  implies  $ C ^* \in \mathcal D$.
On the other hand, $ C^* (u,u)<u$ for all $u\in (0,1)$ by construction. Using Theorem \ref{th:r1-1}, we know that $ C^*$ is not in a bivariate concentration class. This yields a contradiction.
 \end{proof}

 \begin{remark}
 In our formulation of [Cx],
convexity of $\mathcal D$  is required to hold for countable sums. This property is used in the last step of the proof of Proposition \ref{prop:r1-2} to yield that $C^*$ is in $\mathcal D$. The current proof techniques do not work if we  require convexity of $\mathcal D$ only for finite sums.
  \end{remark}

 To conclude the paper,
we put Proposition \ref{prop:r1-2} and Theorem \ref{th-ES} together to arrive at a complete endogenous reasoning for a regulator to use ES.

   \begin{theorem}\label{th:r1-2}
 A monetary risk measure $\rho$ on $L^\infty$ satisfies  lower semicontinuity, $\rho(0)=0$, and $\mathcal D$-aversion for a maximal    bivariate concentration class $\mathcal D$
if and only if $\rho={\rm ES}_p$ for some $p\in (0,1)$.
\end{theorem}
 \begin{proof}
 This result follows from combining
  Theorem \ref{th-ES} and Proposition \ref{prop:r1-2},
  and noting that $\mathcal D_p$-aversion is equivalent to [$p$-CA].
   \end{proof}
 To interpret Theorem \ref{th:r1-2}, we make the following economic assumptions on the regulator's preference towards dependence structures in risk aggregation. First, the regulator believes that breaking LLN is dangerous ([ND]); second, the regulator believes that more positive dependence is more dangerous ([DM]); third, the regulator believes that a mixture of dangerous structures is dangerous ([Cx]); fourth, the regulator chooses to use a largest possible set to model such dangerous structures (maximality).
If all four   assumptions are met, then,  by Proposition \ref{prop:r1-2}, the regulator needs to use a risk measure that satisfies concentration aversion.
With some other standard properties in  Theorem \ref{th-ES},   we further arrive at the class of ES.
Certainly, the desirability of the four assumptions on the regulator's dependence preference can be debated, and, based on the main results of this paper, such debates can directly translate to critical arguments for or against the use of ES in financial regulation.

\begin{remark}
If maximality is removed from the consideration of the regulator, then we can allow for other bivariate concentration classes. The simplest  such example  is the singleton $\mathcal D^+:=\{C^+\}$, which is clearly also the smallest bivariate concentration class.
As shown by \cite{MW20}, a monetary risk measure is $\mathcal D^+$-averse if and only if it is SSD-consistent. 
Therefore, by Theorem \ref{th:r1-2}, the maximality of the bivariate concentration class $\mathcal D$ pins down the class of ES among all    lower-semicontinuous consistent risk measures.
This shows that, among a general class of risk measures,  ES has the largest spectrum of dangerous dependence. 
In other words, if maximality of $\mathcal D$ is desirable, then ES is the only suitable class; if maximality is relaxed to somewhere between the largest and the smallest, then the regulator has more choices of SSD-consistent risk measures.
A larger set of dangerous dependence narrows down the corresponding choices of regulatory risk measures, from all SSD-consistent ones to ES.
\end{remark}

\subsubsection*{Acknowledgements}
The authors thank the Editor, an Associate Editor, two anonymous referees,   Nazem Khan and Yi Shen for helpful comments on an early version of the paper.
The authors would like to thank Martin Herdegen for raising the question of whether the NRC axiom  of \cite{WZ21} can be replaced by an alternative natural property without imposing an equality.
Ruodu Wang acknowledges financial support from  the
Natural Sciences and Engineering Research Council of Canada (RGPIN-2018-03823, RGPAS-2018-522590).

\appendix

\section{Proof of Lemma \ref{lem:r1-1}}
\label{app:A}

\begin{proof}[Proof of Lemma \ref{lem:r1-1}]
Let $C'= C/2+ C^+/2$.
Note that $C^+\ge C$, and hence $C'\ge C$.
Moreover,  since $t_{C^+}=s_{C^+}$ is the identity on $(0,1)$, we have  $t_{C'}=\min (t_C,t_{C^+}) = t_C$
and $s_{C'}=\max (s_C,s_{C^+}) =s_C$.
See Figure \ref{fig:r1-1}a-b for an illustration of $C$ and $C'$.
It suffices to show that there exists $\widetilde C\ge C'$ such that $\widetilde C$ has positive density on $B_C$.

Below, for simplicity, we will use the notation $C$ for $C'$ above. (Alternatively, we can directly assume that the measure $C^+$ is absolutely continuous with respect to the measure $C$, and the above argument guarantees that this assumption is without loss of generality.)

Take $(U,V)\sim C$ and $(U',V')\sim C$ such that $(U,V)$ and $(U',V')$ are independent.  Construct  a random variable $ \widetilde V  $ by
\begin{align}
\label{eq:loop}
 \widetilde V =  V  \id_{A^c} + V'\id_A,~\mbox{where $A=\{U>U',~V<V'\}\cup\{U<U',~V>V'\}$.}
\end{align}
 Note that since $(U,V)$ and $(U',V')$ are iid,  we have,
 for $v\in (0,1)$,
\begin{align*} \p(V'\le v,~ A) &=  \p(V'\le v,~ U>U',~V<V' )  + \p(V'\le v, ~U<U',~V>V' )\\   &= \p(V\le v, ~U'>U,~V'<V )  + \p(V\le v, ~U'<U,~V'>V )
= \p(V\le v, ~A) .\end{align*}
Hence,
\begin{align*}
\p(\widetilde V\le v ) = \p(V\le v, A^c) + \p(V'\le v, A) = \p(V\le v, A^c) + \p(V\le v, A)  =\p(V\le v)= v
 \end{align*} implying that  $\widetilde V$ is uniformly distributed on $[0,1]$.
 As a consequence, the distribution of $(U,\widetilde V)$ is a copula, and we denote it by $\widehat C$.

 We first verify $\widehat C\ge C$.
 Using the fact that $(U,V)$ and $(U',V')$ are iid, we get, for    $(u,v)\in [0,1]^2$,
 \begin{align*}
&\p((U,V')\le (u,v) , ~U<U',~V>V') -  \p((U,V)\le  (u,v), ~U<U',~V>V')
 \\&=  \p( U\le u,~V'\le v< V , ~U<U')
  \\&=   \p( U'\le u,~V\le v< V' , ~U'<U)
    \\&\ge   \p(U\le u,~V\le v <V',~U>U')
 \\& =  \p((U,V)\le (u,v) , ~U>U',~V<V') -  \p((U,V')\le  (u,v), ~U>U',~V<V').
 \end{align*}
As a consequence,
$$
\p( (U,V')\le (u,v) , ~A) \ge \p( (U,V)\le (u,v) , ~A),
$$
and hence $$\p((U,\widetilde V)\le (u,v) ) =
\p( (U,V)\le (u,v) , ~A^c )+
\p( (U,V')\le (u,v) , ~A)
\ge \p( (U,V)\le (u,v) ),
 $$
 which gives the order
$\widehat C(u,v)\ge C(u,v)$. Intuitively, this is because $\widehat C$ is obtained from $C$ via a continuum of CI transfers (see Figure \ref{fig:r1-1}b). 

Finally, we verify the statement on the positive density on $B_C$. For $(s,t)\in [0,1]^2$,
\begin{align*}
\widehat C(s,t) &=\E\left[ \p\left ((U,\widetilde V)\le (s,t) \mid  U,V \right )\right]\\ & =\int_{[0,1]^2} \p\left ((U,\widetilde V)\le (s,t)  \mid (U,V)=(u,v)\right) \d C(u,v)\\ & =\int_{[0,s]\times [0,1]} \p\left (  \widetilde V \le t  \mid (U,V)=(u,v)\right) \d C(u,v)
\\& =
\int_{[0,s]\times [0,1]}\left (    \p\left (  V' \le t ,~ U'> u,~V'< v \right) +   \p\left (  V' \le t  ,~ U'< u,~V'> v \right)\right) \d C(u,v)
\\ &\quad \quad + \int_{[0,s]\times [0,t]}\left (    \p\left (  U'\le  u,~V'\le  v \right) +   \p\left (    U' \ge  u,~V'\ge  v \right)\right) \d C(u,v).
\end{align*}
Write $t\wedge v= \min(t,v)$.
We have  $\widehat C= F + G $ where
$$
F(s,t)=
 \int_{[0,s]\times [0,1]}\left (    \p\left (  V' \le t\wedge v ,~ U'> u \right) +   \p\left (  v<V' \le t  ,~ U'< u   \right)\right) \d C(u,v),$$ and
$$ G(s,t)= \int_{[0,s]\times [0,t]}\left (    \p\left (  U'\le  u,~V'\le  v \right) +   \p\left (    U' \ge  u,~V'\ge  v \right)\right) \d C(u,v).
$$
Note that
  $F$ and $G$ are the distribution functions of two Borel measures on $[0,1]^2$.
  Below we will show that $F$ has  a positive density on  a subset of $B_C$ (this is shown in Figure \ref{fig:r1-1}c), and then we use another construction to obtain positive density on $B_C$.

Since $V'$ is uniform on $[0,1]$, we know   $\p(V'\le t\wedge v)=t\wedge v$, and hence
\begin{align*}
F(s,t) & =
 \int_{[0,s]\times [0,1]}\left (   t\wedge v - \p\left (  V' \le t\wedge v ,~ U'<u \right) +   \p\left (  v<V' \le t  ,~ U'< u   \right)\right) \d C(u,v)
\\&  =
 \int_{[0,s]\times [0,1]}\left (   t\wedge v   +   C(u,t) - 2 C(u, t \wedge v ) \right) \d C(u,v).
\end{align*}
Using (see Section 2.12 of \cite{J14})
\begin{equation}\label{eq:conditionalcopula}   C _{2|1}(v|u)= \frac{\partial    C}{\partial u}(u,v) \mbox{~~~and~~~}   C _{1|2}(u|v)= \frac{\partial    C}{\partial v}(u,v) ~~~~\mbox{almost everywhere},\end{equation}
we get
\begin{align*}
\frac{\partial  F}{\partial s}(s,t)  &  =    \int_{  [0,1]}\left (   t\wedge v   +   C(s,t) - 2 C(s, t \wedge v) \right)   C_{2|1} ( \d v | s) .  
\end{align*}
Exchanging the order of the derivative and the integral (guaranteed by the dominated convergence theorem), we get
\begin{align*}
\frac{\partial^2 F}{\partial s \partial t }(s,t)   & =   \int_{ [0,1]} \left ( \left (   1 - C_{1|2} (s|t)  \right) \id_{\{t<v\}}+  C_{1|2} (s|t)   \id_{\{t>v\}} \right)  C_{2|1} ( \d v | s)
\\&   =     \left (   1 - C_{1|2} (s|t)  \right) \left (   1-C_{2|1} ( t | s)\right) +  C_{1|2} (s|t)  C_{2|1} ( t  | s) .
\end{align*}
Therefore,
we have
$
\frac{\partial^2 F}{\partial s \partial t }(s,t)  >0
$
as soon as $(C_{1|2} (s|t), C_{2|1} (t | s))$ is not $(0,1)$ or $(1,0)$.
Equivalently,
$
\frac{\partial^2 F}{\partial s \partial t }(s,t)  >0
$
if the support of $C$ includes
either  $(s',t)$ to the left of $(s,t)$
and  $(s,t')$ below $(s,t)$,
or   $(s',t)$ to the right of $(s,t)$
and  $(s,t')$ above $(s,t)$.
Since the support of $C$ includes the diagonal line,
 $F$ has positive density on the set
$$
B^*_{C}=\{(s,t): t^*_C(s) \le t \le s^*_C(s)\},
$$
where $t^*_C(s)$ is the essential infimum of $C_{2|1}(\cdot|s)$ and
$s^*_C(s)$ is the essential supremum of $C_{2|1}(\cdot|s)$.
Since $\widehat C=F+G$,
we obtain that $\widehat C$ has positive density on $B^*_{C}$, possibly plus a non-absolutely continuous component coming from $G$.

Generally, the set $B^*_{C}$ may  be different from the set $B_C$. To obtain a positive density on $B_C$, we apply the above procedure again with $\widehat C$ in place of $C$, starting from \eqref{eq:loop}.
This time, we arrive at a new copula $\widetilde C \ge \widehat C$ such that $\widetilde C$ has positive density on $B^*_{\widehat C}$ (this is shown in Figure \ref{fig:r1-1}d).
We claim $B_C\subseteq B^*_{\widehat C}$.
To show this, take $(s,t)\in B_C$.
Assume $t\ge s$, and the case $s<t$ is symmetric.
By definition of $B_C$,
there exists $(s',t')$ in the support of $C$ such that $s' \le s$ and $t' \le t$.
This implies that $(s',t) \in B^*_{C}$.
Since $\widehat C$ has positive density on $B^*_{C}$,
we know that $(s',t)$ and $(s,s)$ are both in the support of $\widehat C$.
This further implies   $(s,t)\in B^*_{\widehat C}$,
and hence $ B_C\subseteq B^*_{\widehat C}$.
Therefore, we conclude that $\widetilde C$ has positive density on $B_C$.
\end{proof}


\begin{thebibliography}{10}

\small




\bibitem[\protect\citeauthoryear{Artzner et al.}{Artzner et al.}{1999}]{ADEH99}
{Artzner, P., Delbaen, F., Eber, J.-M. and Heath, D.} (1999). Coherent measures of risk. \emph{Mathematical Finance}, \textbf{9}(3), 203--228.



\bibitem[\protect\citeauthoryear{Basak and Shapiro}{2001}]{BS01} Basak, S. and Shapiro, A. (2001). Value-at-Risk based risk management: Optimal policies and asset prices. \emph{The Review of Financial Studies}, \textbf{14}(2), 371--405.

 \bibitem[\protect\citeauthoryear{{Basel Committee on Banking
  Supervision}}{{BCBS}}{2016}]{B16}
{BCBS} (2016).
  {\em  Minimum Capital Requirements for Market Risk.  January 2016.}
 Basel Committee on Banking
  Supervision. Basel: Bank for International Settlements. \texttt{https://www.bis.org/bcbs/publ/d352.htm}

    \bibitem[\protect\citeauthoryear{{Basel Committee on Banking
  Supervision}}{{BCBS}}{2019}]{B19}
{BCBS} (2019).
  {\em  Minimum Capital Requirements for Market Risk.  February 2019.}
 Basel Committee on Banking
  Supervision. Basel: Bank for International Settlements. \texttt{https://www.bis.org/bcbs/publ/d457.htm}

\bibitem[\protect\citeauthoryear{Bellini et al.}{2014}]{B14}
Bellini, F., Klar, B., M\"{u}eller, A. and Gianin, E. R. (2014). Generalized quantiles as risk measures. \emph{Insurance: Mathematics and Economics}, \textbf{54}, 41--48.




\bibitem[\protect\citeauthoryear{Chambers}{2009}]{C091}
Chambers, C. P. (2009). An axiomatization of quantiles on the domain of distribution functions. \emph{Mathematical
Finance}, \textbf{19}(2), 335--342.



\bibitem[\protect\citeauthoryear{Dan\'ielsson et al.}{Dan\'ielsson et al.}{2001}]{DEGKMRS01}
Dan\'ielsson, J., Embrechts, P., Goodhart, C., Keating, C., Muennich, F., Renault, O. and Shin, H. S. (2001). An academic response to Basel II.
\emph{LSE Special Paper Series May 2001.}


\bibitem[\protect\citeauthoryear{Delbaen}{2002}]{D02}
Delbaen, F. (2002).   Coherent risk measures on general probability spaces. In \emph{Advances in Finance and Stochastics} (pp. 1--37). Springer, Berlin, Heidelberg.




\bibitem[\protect\citeauthoryear{Embrechts et al.}{Embrechts et al.}{2021}]{EMWW21}
 {Embrechts, P., Mao, T., Wang, Q. and Wang, R.} (2021). Bayes risk, elicitability, and the Expected Shortfall. \emph{Mathematical Finance}, \textbf{31}, 1190--1217.




\bibitem[\protect\citeauthoryear{Embrechts et al.}{Embrechts et al.}{2022}]{ESW21} {Embrechts, P., Schied, A. and Wang, R.} (2022). Robustness in the optimization of risk measures. \emph{Operations Research}, \textbf{70}(1), 95--110.




\bibitem[\protect\citeauthoryear{F\"ollmer and Schied}{F\"ollmer and Schied}{2002}]{FS02} F\"ollmer, H.~and Schied, A.~(2002). Convex measures of risk and trading constraints. \emph{Finance and Stochastics}, \textbf{6}, 429--447.

 \bibitem[\protect\citeauthoryear{F\"ollmer and Schied}{F\"ollmer and Schied}{2016}]{FS16} F\"ollmer, H.~and Schied, A.~(2016). \emph{Stochastic Finance. An Introduction in Discrete Time}. Fourth Edition.  {Walter de Gruyter, Berlin}.


\bibitem[\protect\citeauthoryear{Furman et al.}{Furman et al.}{2017}]{FWZ17}
{Furman, E., Wang, R. and Zitikis, R.} (2017). Gini-type measures of risk and variability: Gini shortfall, capital allocation and heavy-tailed risks. \emph{Journal of Banking and Finance}, \textbf{83}, 70--84.



\bibitem[\protect\citeauthoryear{Gaivoronski and Pflug}{2005}]{GP05}
Gaivoronski, A. and Pflug, G. (2005). Value-at-Risk in portfolio optimization: Properties and computational approach. \emph{Journal of Risk}, \textbf{7}(2), 1--31.




\bibitem[\protect\citeauthoryear{Guo et al.}{2020}]{GKWW20}
Guo, N., Kou, S., Wang, B. and Wang, R. (2020). Self-consistency, subjective pricing, and a theory of credit rating. \emph{SSRN}: 3504065.




\bibitem[\protect\citeauthoryear{Hadar and Russell}{1969}]{HR69}
Hadar, J. and Russell, W. (1969). Rules for ordering uncertain prospects. \emph{American Economic Review},
\textbf{59}(1), 25--34.

		\bibitem[\protect\citeauthoryear{Hampel}{Hampel}{1971}]{H71}
Hampel, F. (1971).
  {A general qualitative definition of robustness}.  \emph{Annals of Mathematical Statistics}, \textbf{42}(6), 1887--1896.


\bibitem[\protect\citeauthoryear{He and Peng}{2018}]{HP18}
He, X. and Peng, X. (2018). Surplus-invariant, law-invariant, and conic acceptance sets must be the sets
induced by Value-at-Risk. \emph{Operations Research}, \textbf{66}(5), 1268--1276.



\bibitem[\protect\citeauthoryear{Herdegen and Khan}{2022}]{HK21}
Herdegen, M. and Khan, N. (2022). Mean-$\rho$ portfolio selection and $\rho$-arbitrage
for coherent risk measures. \emph{Mathematical Finance},  \textbf{32}(1), 226--272.

\bibitem[\protect\citeauthoryear{Joe}{Joe}{2014}]{J14}
Joe, H. (2014).
{\em Dependence Modeling with Copulas}.
London: Chapman \& Hall.








\bibitem[\protect\citeauthoryear{Kou and Peng}{2016}]{KP16}
Kou, S. and Peng, X. (2016). On the measurement of economic tail risk. \emph{Operations Research}, \textbf{64}(5), 1056--
1072.

\bibitem[\protect\citeauthoryear{Kusuoka}{2001}]{K01} Kusuoka, S. (2001). On law invariant coherent risk measures. \emph{Advances in Mathematical Economics}, \textbf{3}, 83--95.

\bibitem[\protect\citeauthoryear{Lin et al.}{2021}]{L21}
Lin, H., Saunders, D. and Weng, C. (2021). Mean-expectile portfolio selection. \emph{Applied Mathematics and Optimization}, \textbf{83}, 1585--1612.



\bibitem[\protect\citeauthoryear{Liu and Wang}{2021}]{LW21}
Liu, F. and Wang, R. (2021). A theory for measures of tail risk.  \emph{Mathematics of Operations Research},  \textbf{46}(3), 1109--1128.




\bibitem[\protect\citeauthoryear{Mao and Wang}{2020}]{MW20}
{Mao, T. and Wang, R.} (2020).  Risk aversion in regulatory capital calculation. \emph{SIAM Journal on Financial Mathematics}, \textbf{11}(1), 169--200.

\bibitem[\protect\citeauthoryear{Marinacci and Montrucchio}{2004}]{MM04}
Marinacci, M. and Montrucchio, L. (2004). Introduction to the mathematics of ambiguity. In \emph{Uncertainty in
Economic Theory} (I. Gilboa, Ed.) 46--107. Routledge, New York.


\bibitem[\protect\citeauthoryear{Markowitz}{Markowitz}{1952}]{M52}
Markowitz, H. (1952). Portfolio selection. \emph{Journal of Finance}, \textbf{7}(1), 77--91.


\bibitem[\protect\citeauthoryear{M\"uller and Stoyan}{M\"uller and Stoyan}{2002}]{MS02} {M\"uller, A. and Stoyan, D.} (2002). \emph{Comparison Methods for Stochastic Models and Risks}. Wiley, England.


\bibitem[\protect\citeauthoryear{Rothschild and Stiglitz}{1970}]{RS70}
Rothschild, M. and Stiglitz, J. (1970). Increasing risk I. A definition. \emph{Journal of Economic Theory}, \textbf{2}(3),
225--243.


\bibitem[\protect\citeauthoryear{Rockafellar and Uryasev}{2000}]{RU00}
Rockafellar, R. T. and Uryasev, S. (2000). Optimization of conditional value-at-risk. \emph{Journal of Risk}, \textbf{2}(3), 21--41.



\bibitem[\protect\citeauthoryear{Rockafellar and Uryasev}{2002}]{RU02}
Rockafellar, R. T. and Uryasev, S. (2002). Conditional value-at-risk for general loss distributions. \emph{Journal of
Banking and Finance}, \textbf{26}(7), 1443--1471.


\bibitem[\protect\citeauthoryear{Rockafellar et al.}{Rockafellar et al.}{2006}]{RUZ06}
Rockafellar, R. T., Uryasev, S. and Zabarankin, M. (2006). Generalized deviation in risk analysis. \emph{Finance and Stochastics}, \textbf{10}, 51--74.




\bibitem[\protect\citeauthoryear{Schmeidler}{1989}]{S89}
Schmeidler, D. (1989). Subjective probability and expected utility without additivity. \emph{Econometrica}, \textbf{57}(3), 571--587.

\bibitem[\protect\citeauthoryear{Tchen}{Tchen}{1980}]{T80}
Tchen, A.~H. (1980).
  Inequalities for distributions with given marginals.
 {\em Annals of Probability},~{\textbf 8}(4), 814--827.

\bibitem[\protect\citeauthoryear{Tierney}{Tierney}{1996}]{T96}
Tierney, L. (1996). Introduction to general state-space Markov chain theory. In \emph{Markov Chain Monte-Carlo in Practice} (W. Gilks,
S. Richardson, and D. Spiegelhalter, Eds), 59–74. Chapman and Hall, London.

\bibitem[\protect\citeauthoryear{Vovk and Wang}{Vovk and Wang}{2021}]{VW21}
Vovk, V. and  Wang, R. (2021).
E-values: Calibration, combination, and applications.  \emph{Annals of Statistics}, \textbf{49}(3), 1736--1754.


\bibitem[\protect\citeauthoryear{Yaari}{1987}]{Y87}Yaari, M. E. (1987). The dual theory of choice under risk. \emph{Econometrica}, \textbf{55}(1), 95--115.

\bibitem[\protect\citeauthoryear{Wang and Wei}{Wang and Wei}{2020}]{WW20b} Wang, R. and Wei, Y. (2020). Risk functionals with convex level sets. \emph{Mathematical Finance},  \textbf{30}(4), 1337--1367.

\bibitem[\protect\citeauthoryear{Wang and Wu}{2020}]{WW20}
 {Wang, R. and Wu, Q.} (2020).  Dependence and risk attitudes: An equivalence.
 \emph{SSRN}: 3707709.

\bibitem[\protect\citeauthoryear{Wang and Zitikis}{2021}]{WZ21}
 Wang, R. and Zitikis, R. (2021). An axiomatic foundation for the Expected Shortfall. \emph{Management Science}, \textbf{67}(3), 1413--1429. \end{thebibliography}
\end{document}